\documentclass[table,11pt,a4paper]{scrartcl}

\usepackage{amssymb,amsmath,amsthm,fullpage,enumerate,url}
\usepackage{multirow,mathtools,booktabs}
\usepackage{hyperref}
\usepackage{tikz}

\usepackage{subcaption}

\newcommand{\axyz}[4][\alpha]{#1_{#2,#3}^{#4}}
\newcommand{\e}{\mathrm{e}}

\newcommand{\dmax}{d_{\mathsf{max}}}
\newcommand{\davg}{d_{\mathsf{avg}}}
\newcommand{\dmin}{d_{\mathsf{min}}}
\newcommand{\ct}{\mathcal{T}}

\renewcommand{\Pr}[1]{\mathbb{P}\left[\,#1\,\right]}
\newcommand{\Pru}[2]{\mathbb{P}_{#1}\!\left[#2\right]}
\newcommand{\Ex}[1]{\mathbb{E} \left[\,#1\,\right]}
\newcommand{\Exu}[2]{\mathbb{E}_{#1} \left[\,#2\,\right]}

\newcommand{\BO}[1]{\mathcal{O}\!\left(#1\right)} 
\newcommand{\BT}[1]{\Theta\!\left(#1\right)}
\newcommand{\lo}[1]{o\!\left(#1\right)}

\newcommand*{\abs}[1]{\lvert #1\rvert}

\newcommand*{\bfrac}[2]{\genfrac{(}{)}{}{}{#1}{#2}}

\newcommand{\Htwo}[2]{H_{#1}^{\mathsf{two}}(#2)}
\newcommand{\Ctwo}[2]{C_{#1}^{\mathsf{two}}(#2)}
\newcommand{\ttwo}[1][cov]{t_{\mathsf{#1}}^{\mathsf{two}}}
\newcommand{\ttwoh}[1][hit]{t_{\mathsf{#1}}^{\mathsf{two}}}

\newcommand{\thit}{t_{\mathsf{hit}}}
\newcommand{\trel}{t_{\mathsf{rel}}}

\newcommand{\mdeg}[1]{d\!\left(#1 \right)}

\newcommand{\dist}{\operatorname{dist}}
\newcommand{\diam}{\operatorname{diam}}

\newcommand{\poly}{\mathop{\mathsf{poly}}}
\newcommand{\polylog}{\mathop{\mathsf{polylog}}}
\newcommand{\NP}{\mathsf{NP}}
\newcommand{\EXP}{\mathsf{EXP}}
\newcommand{\PSPACE}{\mathsf{PSPACE}}

\renewcommand{\leq}{\leqslant}
\renewcommand{\geq}{\geqslant}

\renewcommand{\tilde}{\widetilde}

\newcommand{\Z}{{\mathbb{Z}}}
\newcommand{\eps}{\varepsilon}

\allowdisplaybreaks

\newtheorem{theorem}{Theorem}[section]
\newtheorem{lemma}[theorem]{Lemma}
\newtheorem{corollary}[theorem]{Corollary}
\newtheorem{proposition}[theorem]{Proposition}

\newtheorem{definition}[theorem]{Definition}
\newtheorem*{claim}{Claim}
\newtheorem*{remark}{Remark}
\newenvironment{poc}{\begin{proof}[Proof of claim]}{\end{proof}}
\newenvironment{poci}{\begin{proof}[Proof of Claim \eqref{sum 0}]}{\end{proof}}
\newenvironment{pocii}{\begin{proof}[Proof of Claim \eqref{sum i}]}{\end{proof}}

\newcommand{\Lr}[1]{Lemma~\ref{#1}}

\newcommand{\Tr}[1]{Theorem~\ref{#1}}

\newcommand{\Sr}[1]{Section~\ref{#1}}

\newcommand{\Prr}[1]{Pro\-position~\ref{#1}}

\newcommand{\Cr}[1]{Corollary~\ref{#1}}

\newcommand{\N}{\ensuremath{\mathbb N}}

\usepackage{authblk}

\begin{document}
\providecommand{\keywords}[1]{\textbf{\textit{keywords---}} #1}

\title{The Power of Two Choices for Random Walks\footnotetext{A preliminary version of this paper appeared at \newblock \emph{The 11th  Innovations  in  Theoretical  Computer  Science Conference (ITCS 2020)}, volume 151 of LIPIcs, pages 76:1--76:19 \cite{ITCSpaper}}}
	\author[1]{Agelos Georgakopoulos}
	\author[1]{John Haslegrave}
	\author[2]{Thomas Sauerwald}
	\author[2]{John Sylvester}
	\affil[1]{Mathematics Institute, University of Warwick}
	\affil[2]{Department of Computer Science \& Technology, University of Cambridge}
	\date{\vspace{-5ex}}
\maketitle
\begin{abstract}
	
We apply the power-of-two-choices paradigm to a random walk on a graph: rather than moving to a uniform random neighbour at each step, a controller is allowed to choose from two independent uniform random neighbours. We prove that this allows the controller to significantly accelerate the hitting and cover times in several natural graph classes. In particular, we show that the cover time becomes linear in the number $n$ of vertices on discrete tori and bounded degree trees, of order  $\mathcal{O}(n \log \log n)$ on bounded degree expanders, and of order  $\mathcal{O}(n (\log \log n)^2)$ on the Erd\H{o}s-R\'{e}nyi random graph in a certain sparsely connected regime. 

We also consider the algorithmic question of computing an optimal strategy, and prove a dichotomy in efficiency between computing strategies for hitting and cover times.
\end{abstract}
\setcounter{footnote}{1}
\section{Introduction}

The power of choice paradigm asserts that when a random process is offered a choice between two or more uniformly selected options, as opposed to being supplied with just one, then a series of choices can be made to improve the overall performance. This idea was first applied to the `balls into bins' model \cite{AzarBalanced,HeavySTOC,Mitz}, where it was proved that the power of choice decreases the maximum load from $\BT{\frac{\log n}{\log\log n}}$ to $\BT{\log\log n}$ when assigning $n$ balls to $n$ bins.  The power of choice was later extensively studied for random graphs under the broader class of rule-based random graph processes, known as Achlioptas processes, see e.g.\ \cite{SpeArc,BF-Ach,BK-Ach,RioArc,SW-Ach} and references therein. 
The power of choice has also been studied with regard to the Preferential Attachment process for growing a random connected graph; in this context the choices may have a powerful effect on the degree distribution, see e.g.~\cite{HJchoicePA,MPchoicePA}.

In this paper we extend the power of two choices paradigm to random walk on a graph. We show that for many natural classes of graphs this results in a significant speed-up of the cover and hitting times, which are the expected times to visit all vertices or any fixed vertex from a worst case start vertex. 
We study the \emph{choice random walk (CRW)}, which at every step is offered two uniformly random independently sampled neighbours (with repetition) of the current location, and (with full knowledge of the graph) must choose one as the next step; see Section \ref{formaldef} for more details. We prove that the cover time of CRW decreases to $\BT{n}$ for grids (i.e.\ finite quotients of $\Z^d$) and bounded degree trees on $n$ vertices, and that the cover time of expander graphs decreases to $\BO{n \log\log n}$. We note that for simple random walk (SRW) these cover times are all $\Omega(n\log n)$ and some are $\Omega(n^2)$ \cite{aldousfill}. We also consider computational questions relating to choosing a good strategy: we show that an optimal strategy for minimising a hitting time can be computed in polynomial time, but choosing an optimal strategy for minimising the cover time is $\NP$-hard. See \Sr{sec res} for more details and other results.

Part of our motivation is to improve the efficiency of random walks used in algorithmic applications such as searching, routing, self-stabilization, and query processing in wireless networks, peer-to-peer networks and other distributed systems. One practical setting where routing using the power of choice walk may be advantageous is in relatively slowly evolving dynamic networks such as the internet. For example, say a packet has a target destination $v$ and each node stores a pointer to a neighbour which it believes leads most directly to $v$. If this network is perturbed then the deterministic scheme may get stuck in ``dead ends'' whereas a random walk would avoid this fate. The CRW which prefers edges pointed to by a node may be the best of both worlds as it would also avoid traps but may see a speed up over the simple random walk when the original paths are still largely intact.  

\subsection{Related Literature}
To the best of our knowledge, Avin \& Krishnamachari~\cite{AvinKri} were the first to apply the principle of the power of choice to random walks. However, their version only considers a simple choice rule where the vertex with fewer previous visits is always preferred, and ties are broken randomly. This is in the spirit of balanced allocations, the origin of the power of two choices paradigm. Their results are mainly empirical and suggest a decrease in the variance of the cover time, and a significant improvement in visit load balancing. This is related to the greedy random walk of Orenshtein and Shinkar~\cite{OreShi}, which chooses uniformly from adjacent vertices that have not yet been visited (if possible). This model is well studied for expanders~\cite{BerUnVisit,CooBia}. The power of choice has also been studied in the context of deterministic random walks and the rotor-router model \cite{BeeDet,CooperICALP}.  

Perhaps closest to our work, Azar, Broder, Karlin, Linial and Phillips \cite{ABKLPbias} introduced the $\eps$-biased random walk ($\eps$-BRW) where at each step with probability $\eps>0$ a controller can choose a neighbour of the current vertex to move to, otherwise a uniformly random one is selected. The model is quite similar to ours in the sense that the controller has full knowledge of the graph when choosing a neighbour. Azar et al.\ obtained bounds on the stationary probabilities and show that optimal strategies for maximising or minimising stationary probabilities or hitting times can be computed in polynomial time. There is some overlap with our results in Section \ref{complexsec}, where in particular Theorem \ref{progstat} uses a clever substitution from \cite{ABKLPbias} to express an optimisation problem as a linear program. One major difference is that Azar et al.\ restrict their study to time independent strategies and do not investigate cover times. Three of the authors of this paper have recently extended \cite{ABKLPbias} to the time dependent setting and studied cover times for $\eps$-biased random walks \cite{TBRW}. The conference paper \cite{ITCSpaper} collects some of our results on the CRW from here and on the $\eps$-BRW from \cite{TBRW} giving a comparison between the two processes. 

Azar et al.\ \cite{ABKLPbias} suggest that the most natural choice of bias for the $\eps$-BRW is $\eps=\Theta(1/\dmax)$, where $\dmax$ is the maximum degree. It is shown in \cite[Proposition 1]{ITCSpaper} that the CRW can emulate the $\eps$-BRW provided $\eps\leq 1/\dmax$. However, the reverse does not hold unless the bias $\eps$ is close to $1$; the main obstacle is that avoiding a particular next step is much more difficult for the $\eps$-BRW.
Further evidence that the CRW is more powerful than the $\eps$-BRW is in the cover time bounds we prove for the CRW in Theorem \ref{trelbdd} and for the time dependent version of the $\eps$-BRW in \cite[Theorem 3.2]{TBRW}. For the most natural choice $\eps=\Theta(1/\dmax)$, these bounds differ by a factor which is almost linear in $\dmax$, suggesting that the CRW deals better with high-degree graphs than the $\eps$-BRW.

With regard to complexity questions, we note that for the simple random walk, hitting times can be expressed as the solution to a set of $n$ linear equations and can therefore  be computed in polynomial time. Determining the complexity of computing the cover time, however, is far more challenging and still remains open \cite[Open Problem 6.35]{aldousfill}. Significant progress was made by Ding, Lee and Peres \cite{DLP} who discovered a deterministic polynomial time $\BO{1}$-approximation algorithm for the cover time. In this paper we show that computing an optimal strategy for the cover time of the CRW is $\NP$-hard. 

\subsection{Our Results} \label{sec res}

In this section we shall present the main results we have obtained for CRW. The numbers of these theorems correspond to where they appear in the paper, although some theorem statements have been simplified for ease of exposition. 

The CRW is not reversible in general, however, we show that it can emulate certain reversible chains.  Combining this with the well-known  connection between electrical networks and reversible Markov chains, we obtain the following general bound on the maximum hitting time $\ttwoh(G)$ between any two vertices of a graph $G$. 
\begin{theorem}\label{quad-cover}For any finite graph $G$ we have $\ttwoh(G)<\min\{3\abs{E}, n^2\}$.
\end{theorem}

This is tight up to constants at both ends of the density spectrum and improves considerably over the well-known $\BO{n \abs{E}}$ worst-case bound for the simple random walk. A witness to tightness for sparse graphs is traversing a path from end to end, 
and for dense graphs hitting a vertex connected by a single edge to a clique.       

\medskip
Most of this paper focuses on the cover time $\ttwo(G)$ for CRW on a graph $G$  under an optimal strategy. For the SRW $\thit$, the maximum hitting time between any two vertices, determines the cover time up to a $\log n$ factor by Matthew's bound \cite[Ch.\ 11.2]{levin2009markov}. However, due to the effect of the choices, this does not apply to the CRW and so we develop other methods to bound $\ttwo$. 

The next result implies that  $\ttwo(T)$ is linear for a bounded degree tree $T$:

\begin{theorem}\label{tree linear}
	For every $d\in \N$ and every $n$-vertex tree $T$  with maximum degree $d$, we have \[\ttwo\left(T\right)\leq 8dn.\]
\end{theorem}
Our strategy for achieving this changes with time, and covers the vertices of $T$ in a prescribed order.

Next, we obtain a similar result for $d$-dimensional grids and tori. The proof technique is different: we 
show that there exists a CRW strategy for the infinite $d$-dimensional grid under which the CRW becomes strongly recurrent. In particular, the expected crossing time of any edge is finite. We use this to deduce
	\begin{theorem}\label{covergrid}
	For any $d$, and any  $d$-dimensional $n$-vertex torus or grid $G$, we have  $\ttwo(G)=\BT{n}$ and $\ttwoh(G) =\BT{\diam(G)} = \BT{n^{1/d}}$.
	\end{theorem}

Avin \& Krishnamachari~\cite{AvinKri} conjecture a speed up for their aforementioned local power of two choice walk on the $2$-dimensional grid. Theorem \ref{covergrid} corroborates this for our global version of the process, but does not yet prove their  conjecture.

\medskip
We develop a method for boosting the probabilities of rare events in the CRW, which gives bounds on hitting and cover times. Perhaps the most important application of these methods is to expander graphs:

	\begin{theorem}\label{trelbddcor}For every sequence $(G_n)_{n\in \N}$ of bounded degree expanders, where $G_n$ has $n$ vertices, we have  \[\ttwo(G_n)=\BO{n \log \log n}.\] 	\end{theorem}

\Tr{trelbddcor} is in fact an immediate corollary of a more general bound (\Tr{trelbdd}), bounding $\ttwo(G)$ in terms of the hitting time (of the SRW), relaxation time and degree discrepancy. In particular, these bounds apply w.h.p.\ to the random $d$-regular graph for fixed $d$. Another application of these methods gives the following bounds for the Erd\H{o}s--R\'enyi random graph, showing a significant improvement on cover time for the regime with sub-polynomial growth of the average degree.
\begin{theorem}\label{gnp}Let $\mathcal{G} \overset{d}{\sim}\mathcal{G}(n,p)$ where $  np\geq c\ln n$ for any fixed $c>1$ and $\log np =\lo{\log n}$. Then w.h.p.\ 
	\begin{enumerate}[(i)]
		\item  $\ttwo\left(\mathcal{G}\right) = \BO{n \cdot \log(np)\cdot \log\log n}$ 
		\item  $\ttwoh\left(\mathcal{G}\right) = n^{1-\Omega(1/ \log(np))}$ .
	\end{enumerate} 
	\end{theorem}

Finally, \Sr{complexsec}  deals with the computational complexity of computing optimal strategies to minimise hitting and cover times. We show the following dichotomy: an optimal strategy to hit a set of vertices can be computed efficiently whereas choosing between two cover time strategies is $\NP$-hard. 
More precisely, we have

 \begin{theorem}\label{hitexact}For any graph $G$, $S \subset V$ and $x\in V\setminus S$, a strategy minimising the hitting time of $S$ from $x$ can be computed in time $\poly(\abs{V})$. \end{theorem}

Notice that any strategy for covering a graph must specify a set of choice preferences from every vertex for every possible set of vertices covered and thus may have size exponential in $n$. This makes the second half of the dichotomy as phrased above sound somewhat modest. However, what we show is that even in the ``on-line'' setting where one is given the set covered so far then just choosing the next step (outputting something of polynomial size) is $\NP$-hard.
    \begin{theorem}\label{CovIsHard} Given the covered set $X$ and position $v$ of the walk at some time, it is $\NP$-hard to choose the next step from two neighbours of $v$ so as to minimise the expected time for the CRW to visit every vertex not in $X$, assuming an optimal strategy is followed thereafter.
\end{theorem}
Our proof shows that this remains $\NP$-hard if $G$ is constrained to have maximum degree $3$. To the best of our knowledge, this is the first intractability result for processes involving random walks with choice.

\medskip
Our results for fundamental graph topologies are summarised in Table \ref{tbl:results}, along with the corresponding hitting and cover times for the simple random walk for ease of comparison. 

\begin{table}[ht]  
	\resizebox{\textwidth}{!}{\begin{tabular}{@{}lcccc@{}} \toprule
			& \multicolumn{2}{c}{Hitting time} & \multicolumn{2}{c}{Cover time} \\ \cmidrule(lr){2-3} \cmidrule(lr){4-5}
			Graph family & SRW & CRW & SRW & CRW \\ \midrule
			Subcubic ($\dmax\leq 3$) & $\BO{n^2}$ &  $\BT{\diam(G)}$& $\BO{n^2}$ & $\BT n$ \\ \addlinespace[.25\defaultaddspace]
			$2$-dim.\ torus/box  & $ \BT{n\log n}$& $\BT{n^{1/d}}$ & $\BT{n\log^2 n}$ &  $\BT n$  \\ \addlinespace[.25\defaultaddspace]
			$d$-dim.\ torus/box  & $ \BT n$& $\BT{n^{1/d}}$ & $ \BT{n\log n}$ &  $\BT n$  \\ \addlinespace[.25\defaultaddspace]
			Bounded-degree tree & $\BO{n^2}$ & $\BO{n}$ & $\BO{n^2}$ & $\BT n$  \\ \addlinespace[.25\defaultaddspace]
			$d$-regular expander & $\BT{n}$ & 
			$n^{1-\Omega(1/\log d)}$& $\BT{n\log n}$  &$\BO{n\log(d)\log\log n}$ \\ \addlinespace[.25\defaultaddspace]
			$\mathcal{G}(n,\polylog(n)/n)$ & $\BT n$  & $n^{1-\Omega(1/\log\log n)}$  & $\BT{n\log n}$&  $\BO{n(\log\log n)^2}$\\ \addlinespace[.25\defaultaddspace]
			Complete graph & $n-1$ & $\frac{(n-1)^2}{2n-3}$   & $\sim n\ln n$ & $\sim(n\ln n)/2$  \\ \bottomrule
	\end{tabular}}\caption{All graphs are on $n$ vertices. The third and fifth columns contain our results on the hitting and cover times of the choice random walk. The second and fourth columns give corresponding quantities for the simple random walk for comparison, which can be found in \cite{aldousfill} and elsewhere. For random graphs, these bounds apply w.h.p.}
	\label{tbl:results}
\end{table}

\section{Preliminaries} \label{formaldef}

The \textit{choice random walk} (CRW) is a discrete time stochastic process $(X_t)_{t\geq 0}$ on the vertices of a connected graph $G=(V,E)$, influenced by a \emph{controller}. The starting state is a fixed vertex; at each time $t\in\mathbb N$ the controller is presented with two neighbours $\{c_1^t,c_2^t \}$ of the current state $X_t$ chosen uniformly at random with replacement and must choose one of these neighbours as the next state $X_{t+1}$. We assume that at each time $t$ the controller knows the graph $G$, its current position $X_t\in V$, and $\mathcal{H}_t:=\left(X_i,\{c_1^i,c_2^i\}\right)_{i=0}^t$, the \textit{history} of the process so far. The controller has access to arbitrary computational resources and an infinite string of random bits $\omega$ in order to choose $X_{t+1}$ from $\{c_1^t,c_2^t\}$. A \textit{CRW strategy} is a function which given any  $t$, $\mathcal{H}_t$ and $\{c_1^t, c_2^t\} \subseteq \Gamma(X_t) $, outputs one of $\{c_1^t, c_2^t\}$ (where we write $\Gamma(v) := \{w : vw \in E \}$ for the neighbourhood of $v$).  Note that any such strategy defines a Markov chain on $V$. 

We say that a CRW strategy is \textit{unchanging} if it is independent of both time and the history of the walk. We say that an unchanging strategy is \textit{reversible} if the Markov chain it defines is reversible. We recall that any reversible Markov chain is identically distributed with a random walk on an edge-weighted graph as explained e.g.\ in \cite{aldousfill}, we shall make use of this representation. For many graphs with a high degree of symmetry we can find good reversible strategies, and we can then use tools from the theory of reversible Markov chains to analyse the CRW on these graphs. 
The strategies we consider may use random bits in addition to those used for choosing $\{c_1^t, c_2^t\}$; we say a strategy is \textit{deterministic} if no additional random bits are used.

If we are trying to minimise the expected hitting time of a given vertex, it is easy to see that there is an unchanging, deterministic optimal strategy. However, it need not be reversible; an example where it is not is given in Figure \ref{bull}. We shall use reversible strategies to bound the hitting time of the optimal strategy; these will also in general not be deterministic.

\begin{figure}
\centering
\begin{tikzpicture}[xscale=0.9,yscale=0.9,knoten/.style={thick,circle,draw=black,minimum size=.5cm,fill=white},target/.style={thick,circle,draw=black,minimum size=.5cm,fill=red},edge/.style={thick,black},dedge/.style={thick,black,-stealth}]
	\node[knoten] (u) at (0,2.598) {};
	\node[knoten] (v) at (3,2.598) {};
	\node[knoten] (w) at (4.5,0) {};
	\node[knoten] (x) at (6,2.598) {};
	\node[target] (y) at (9,2.598) {};
	\draw[edge] (u) to (v);
	\draw[edge] (v) to (w);
	\draw[edge] (w) to (x);
	\draw[edge] (x) to (y);
	\draw[edge] (v) to (x);
	\draw[dedge] (0.25,2.848) to node[above] {$\scriptstyle 1$} (1.5,2.848);
	\draw[dedge] (3.25,2.848) to node[above] {$\scriptstyle 5/9$} (4.5,2.848);
	\draw[dedge] (6.25,2.848) to node[above] {$\scriptstyle 5/9$} (7.5,2.848);
	\draw[dedge] (2.75,2.348) to node[below] {$\scriptstyle 1/9$} (1.5,2.348);
	\draw[dedge] (8.75,2.348) to node[below] {$\scriptstyle 1$} (7.5,2.348);
	\draw[dedge] (3.342,2.506) to node[right] {$\scriptstyle 3/9$} (3.967,1.424);
	\draw[dedge] (5.75,2.348) to node[below] {$\scriptstyle 1/9$} (4.5,2.348);
	\draw[dedge] (6.092,2.256) to node[right] {$\scriptstyle 3/9$} (5.467,1.174);
	\draw[dedge] (4.158,0.095) to node[left] {$\scriptstyle 1/4$} (3.534,1.174);
	\draw[dedge] (4.405,0.342) to node[left, near end] {$\scriptstyle 3/4$} (5.033,1.424);
	\end{tikzpicture}
	\caption{Optimal CRW transition probabilities on the bull graph for hitting the rightmost vertex. The corresponding Markov chain is not reversible.}\label{bull}
\end{figure}
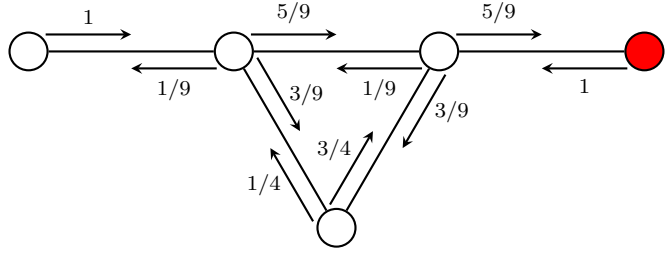

\medskip
For a strategy $\alpha$ and for a vertex $v$ and distinct neighbours $i,j$ let $\axyz vij\in [0,1]$ be the probability that when the walk is at $v$ it chooses $i$ when offered $\{i,j\}$ as choices, i.e.\ $\axyz vij:=\Pr{X_{t+1}=i \mid X_{t}=v, c^t=\{i,j\}}$ (this probability is also conditional on $\mathcal{H}_{t}$ but we suppress this for notational convenience). These are the only parameters we may vary, but we shall find it convenient to define $\axyz vii:=1/2$ for each $i$ adjacent to $v$. Thus 
\begin{equation}\label{alphacond}\text{for each }v\in V: \;\axyz vij \in [0,1] \text{ and } \axyz vji = 1-\axyz vij \text{ for all } i,j \in \Gamma(v).\end{equation}
The transition probabilities $q_{v,i} $ for the strategy $\alpha$ are then given by 
\begin{equation}\label{potctp}q_{v,i}=\frac{2\sum_{j\in \Gamma(v) } \axyz vij}{\mdeg{v}^2}.\end{equation}
For a family of parameters $\axyz vij$ to yield a valid set of transition probabilities $q_{i,j}$ they must satisfy \begin{equation}\label{alphasumcond} \sum_{i \in \Gamma(v)}\sum_{j \in \Gamma(v)}\axyz vij =\frac{\mdeg{v}^2}{2}\end{equation}
for every $v\in V$. Notice that any weights satisfying \eqref{alphacond} also satisfy \eqref{alphasumcond}.

\medskip
Let $\Ctwo{v}{G}$ denote the minimum expected time (taken over all strategies) for the CRW to visit every vertex of $G$ starting from $v$, and define the \textit{cover time} $\ttwo(G):=\max_{v\in V}\Ctwo vG$. Analogously, let $\Htwo xy$ denote the minimum expected time for the CRW to reach $y$, which may be a single vertex or a set of vertices, starting from a vertex $x$ and define the \textit{hitting time} $\ttwoh(G): = \max_{x,y\in V}\Htwo xy$. We drop the superscript from this notation when referring to the associated quantities for the simple random walk.

\section{Bounds from Weighted Graphs}
In this section we analyse CRW strategies which emulate a random walk on a weighted graph. We prove a tight general bound on hitting times and show that any vertex of a graph with maximum degree $3$ can be hit in time proportional to its distance from the start vertex.  
\subsection{An Extremal Hitting Time Result}
In this section we prove that $\ttwo[hit](G) = \BO{e(G)}$ for an arbitrary graph $G$, where $e(G)$ is the number of edges. This bound is best possible up to the implied constants: for sparse graphs, the path has $\ttwo[hit]$ around $2e(G)$. For dense graphs, a clique with a pendant path, where the length of the path is growing much slower than the size of the clique, gives $\ttwo[hit]$ around $3n^2/8$. 
\begin{lemma}Fix a vertex $v$, and partition its neighbours into two sets, $A$ and $B$. There is an unchanging strategy for the CRW
	such that whenever the walker is at $v$ it moves to a random neighbour according to the probability distribution
	in which every vertex in $B$ is twice as likely as every vertex in $A$.\end{lemma}
\begin{proof}
	Fix some number $p\in[0,1]$, and consider the following strategy for moving from $v$. If offered two choices from the same set, choose between them uniformly
	at random, but if offered one choice from $A$ and one choice from $B$, choose the one from $A$ with probability $p$. 
	Clearly all elements of $A$ are equiprobable, as are all elements of $B$, so it is sufficient to show that for some $p$ this strategy chooses
	an element of $A$ with probability $q=\frac{\abs A}{\abs A+2\abs B}$. If this is the case each element of $A$ will be chosen with probability $\frac{1}{\abs{A}+2\abs{B}}$ and each element of $B$ w.p. $\frac{2}{\abs{A}+2\abs{B}}$. 
	If $p=1/2$ the probability of choosing an element of $A$ is $\frac{\abs A}{\abs A+\abs B}\geq q$,
	and if $p=0$ then it is $q':=\bfrac{\abs A}{\abs A+\abs B}^2$. Since $(\abs A+\abs B)^2\geq\abs A(\abs A+2\abs B)$, we have $q'\leq q$, and hence for some 
	$p\in[0,1/2]$ we have the required probability by continuity.
\end{proof}
By considering the strategy at each vertex separately, we immediately get the following consequence.
\begin{corollary}\label{double-weight}
Let $G=(V,E)$ be a locally finite weighted graph with weight function $w: E\to \mathbb{R}_+$, having the property that for any two incident edges $xy,xz$ either $w(xy)=w(xz)$, or $w(xy)=2w(xz)$, or $2w(xy)=w(xz)$. Then there is an unchanging strategy for the CRW on $G$ which simulates the random walk defined by the weights $w$.
\end{corollary}

Here, by the \emph{random walk defined by the weights $w$} we mean the reversible Markov chain where the transition probability from a vertex $x$ to a neighbour $y$ is proportional to $w(xy)$.
For a weighted graph $(G,w)$, write $w(G)=\sum_{e\in E(G)}w(e)$.
\begin{lemma}\label{weighted-hit}Let $(G,w)$ be a finite weighted graph, and let $x$ be a vertex such that every edge incident with $x$ has weight $1$. 
	Then for any vertex $y$ adjacent to $x$, we have \[H_{y}(x)\leq w(G)+w(G\setminus x).\]\end{lemma}
\begin{proof}Since the stationary distribution is given by $\pi_v=\frac{1}{2w(G)}\sum_{u\sim v}w(uv)$,
	we have expected return time to $x$ of $\pi^{-1}_x=\frac{2w(G)}{d(x)}$, see for example \cite[Sect.\ 3.2]{aldousfill}. 
	Thus 
	\[\frac{2w(G)}{d(x)}=1+\sum_{z\sim x}\frac1{d(x)}H_{z}(x),\]
	implying
	\[H_{y}(x)\leq\sum_{z\sim x}H_{z}(x)=2w(G)-d(x)=w(G)+w(G\setminus x).\qedhere\]
\end{proof}
We now restate and prove our result for CRW hitting times.  
\newtheorem*{thm:quad-cover}{Theorem \ref{quad-cover}}
\begin{thm:quad-cover}For any finite graph $G$ we have $\ttwoh(G)<\min\{3\abs{E}, n^2\}$.
\end{thm:quad-cover}

\begin{proof}We have to show that the above bounds apply to $\Htwo yx$ for two arbitrary vertices $x,y$. 
	Define a weight function $w:E(G)\to\mathbb R_+$ by $w(uv)=2^{-\min(d(u,x),d(v,x))}$. Note that $w$ satisfies the requirements of 
	Corollary~\ref{double-weight}, so we can bound $\Htwo yx$ by the corresponding hitting time of the random walk on $(G,w)$. We will now bound
	the latter hitting time.
	
	Write $d$ for the maximum distance of a vertex from $x$, and $V_k$ for the set of vertices at distance exactly $k$ from $x$.
	Note that if $y\in V_{k+1}$ then 
	\[H_{y}(x)\leq H_{y}(V_{k})+\max_{z\in V_{k}}H_{z}(x),\]
	and consequently
	\[\max_{y\in V(G)}H_{y}(x)\leq\sum_{k=0}^{d-1}\max_{z\in V_{k+1}}H_{z}(V_{k}).\] 
	For each $0\leq k\leq d-1$ let $G_k$ be the simple weighted graph obtained by deleting $\bigcup_{i<k}V_i$ and identifying vertices in $V_k$ to give a vertex $v_k$; 
	if a vertex in $V_{k+1}$ has multiple edges to $V_k$, delete all but one of them to leave a simple graph. Since removing
	edges between $V_{k+1}$ and $V_k$ cannot reduce the hitting time of $V_k$, we have for any $z\in V_{k+1}$ that $H^{G}_z(V_k)\leq H^{G_k}_z(v_k)$. Note that
	the latter hitting time is unchanged by multiplying all weights by $2^k$, and since every $z\in V_{k+1}$ is adjacent to $v_k$ in $G_k$,
	by Lemma~\ref{weighted-hit} we have $H^{G_k}_z(v_k)\leq 2^k(w(G_k)+w(G_k\setminus v_k))$.
	Thus 
	\[\max_{y\in V(G)}H_{y}(x)\leq\sum_{k=0}^{d-1}2^k(w(G_k)+w(G_k\setminus v_k)).\]
	If $e$ is an edge between $V_j$ and $V_{j+1}$ then the contribution of $e$ to the $k$th term of the above sum is $2^{k-j+1}$ if $k<j$, at most $1$ if $k=j$ and $0$ otherwise,
	so its total contribution is less than $3$, and is less than $2$ if $e$ is one of the edges deleted to make $G_j$ simple. If $e$ is an edge within $V_j$ then its contribution to the $k$th term is $2^{k-j+1}$ if $k<j$ and $0$ otherwise, so
	its total contribution is less than $2$. The first bound follows. Note that of the edges of the first type which are not deleted, there is exactly one from each vertex (other than $x$) to a vertex in a lower layer of $G$, and so these edges form a tree. Thus there are $n-1$ such edges, whose contribution is bounded by $3$, and at most $\binom n2-(n-1)$ other edges, whose contribution is bounded by $2$, giving a bound of $2\binom n2+n-1=n^2-1$.
\end{proof}

\subsection{Cover Times of Subcubic Graphs}
In this section we prove the CRW cover time of any subcubic graph is linear in the number of vertices, where we remind the reader that a subcubic graph is a graph with maximum degree $3$.

\begin{proposition}\label{binary}Let $G$ be any connected graph of maximum degree $3$. Then $\Htwo{u}{v}\leq 9$ for any $uv\in E(G)$. If in addition $G$ is finite with $n$ vertices then $\ttwo(G) = \BT{n}$.
\end{proposition}
\begin{proof}For each $w\neq v$ choose a neighbour $f(w)$ such that $d(f(w),v)<d(w,v)$. Set $f(v)=u$. Let $(X_t)$ be a CRW starting at $u$ using the strategy of choosing $X_{t+1}=f(X_t)$ whenever possible. Couple this with a random walk $Y_t$ on $\Z$ starting at $1$, with $Y_{t+1}=Y_t-1$ if $X_{t+1}=f(X_t)$ and $Y_{t+1}=Y_t+1$ otherwise. Clearly $Y_t\geq\dist(X_t,v)$ at every step of the walk, and so $X_t$ reaches $v$ on or before the first time $t$ that  $Y_t=0$. Since $\Pr{Y_{t+1}=Y_t-1}\geq 5/9$, we have
	\begin{equation}\label{HYt}\Htwo uv\leq \Ex{\min\{t:Y_t=0\}\mid \{Y_0=1 \} }\leq 9.
	\end{equation}
	
	If $G$ has $n$ vertices, let $v$ be any vertex and choose a spanning walk in the graph starting at $v$ and having at most $2n-3$ edges. Such a walk always exists, for example a depth-first exploration of a spanning tree. Proceed in $2n-3$ rounds, in each round using the strategy above to hit the next vertex of the walk. Each round has expected duration at most $9$ by \eqref{HYt}, and so $\ttwo(G)\leq 18n-27$.\end{proof}
\begin{remark}Since $\ttwoh(G)\leq\ttwo(G)$, this is also linear. Even for $3$-regular graphs the diameter could grow linearly, so this is the best possible.\end{remark}
\section{Trees}

In this section we show that $\ttwo(T)=\Theta(n)$ for trees $T$ of bounded degree. Even more, we will prove that we can even specify an arbitrary (closed) walk $W$ traversing each edge of $T$ once in each direction, and cover the vertices of $T$ in the order dictated by $W$ in linear expected time. This is the gist of the following result:
\newtheorem*{thm:tree linear}{Theorem \ref{tree linear}}
\begin{thm:tree linear}
	For every $d\in \N$ and every tree $T$  with maximum degree $d$, we have \[\sum_{x,y\in V(T), xy\in E(T)} \Htwo xy \leq 8d\abs{V(T)}.\]
\end{thm:tree linear}

This result will be proved by realising a strategy to cover $T$ as a sequence of weighted walks, and then bounding the hitting times in these walks using the Essential Edge Lemma. We shall now remind the reader of the setting and statement of this lemma: We say that an edge $vx$ of a graph is \textit{essential} if its removal would disconnect the graph, into two components $A(v, x)$ and $A(x, v)$, say, containing $v$ and $x$ respectively. Let $E(v, x)$ be the set of edges of $A(v, x)$.
\begin{lemma}[Essential Edge Lemma {\cite[Lem.\ 5.1]{aldousfill}}]\label{essentialedge} Let $G$ be any graph with edge weights $\{w_{i,j} \}_{ij\in E}$. If $vx$ is an essential edge, then
	\begin{equation*}
	H_v(x) = \frac{2\sum_{\{i,j\} \in E(v,x) }w_{i,j}}{w_{v,x}} + 1 , 
	\end{equation*}
	where $H_v(x)$ is the hitting time of the reversible Markov chain defined by these edge weights.
\end{lemma}

We now define the CRW strategies we use in the proof of \Tr{tree linear}. Given a tree $T$, we pick an arbitrary \emph{`root'} vertex $r\in V(T)$. In order to obtain an upper bound on $\Htwo xy$ for $x,y\in V(T)$ such that $ xy\in E(T)$, we follow the (unchanging) strategy $\sigma^{xy}$ making the following choices at each vertex $v$:
\begin{equation}\label{sigma xy}\begin{aligned} 
\text{Reduce the distance to $y$  if possible. Otherwise, choose uniformly}\\  \text{an option that increases distance to $r$ if at least one is available.}
\end{aligned}
\end{equation}
In other words, $\sigma^{xy}$ prefers the unique neighbour $w$ of $v$ with  $d(w,y)<d(v,y)$, avoids the unique neighbour $z$  with  $d(z,r)<d(v,r)$, and is indifferent among all other neighbours of $v$. 

We emphasise that $r$ was an arbitrary vertex, but it is important for our calculations below that it is fixed for all  $\sigma^{xy}, x,y\in V(T)$.

Since the strategy $\sigma^{xy}$ is unchanging, there is an assignment of weights $w_{x,y}(e), e\in E(T)$ such that the corresponding random walk (as defined after \Cr{double-weight}) 
is equidistributed with the CRW under strategy $\sigma^{xy}$ when both walks start at $x$ and stop when first visiting $y$.
These weights can be multiplied by any positive constant without changing the random walk they define, and we normalise by fixing $w_{x,y}(xy)=1$ for concreteness.
The rest of the weights can be calculated explicitly, and so we can apply the Lemma \ref{essentialedge} to give the bound
\begin{equation} \label{w xy e}
\Htwo xy< 2\sum_{e\in E(T)} w_{x,y}(e),
\end{equation}
with the understanding that we set $w_{x,y}(e)=0$ if $y$ separates $x$ from $e$ as this edge does not contribute to the sum in Lemma \ref{essentialedge}.

The latter formula expresses $\Htwo xy$ as a sum of contributions of each $e\in E(T)$. The main surprise in the proof of \Tr{tree linear} is the following lemma, which says that for each $e\in E(T)$, the sum of these  contributions $w_{x,y}(e)$ over all $\Htwo xy, x,y\in V(T), xy\in E(T)$, is bounded. An obvious double counting argument involving \eqref{w xy e} will then establish \Tr{tree linear}. 

\begin{lemma}\label{edge bound}
	For every $d\in \N$, every tree $T$ with maximum degree $d$, and every edge $e\in E(T)$, we have $\sum_{x,y\in V(T), xy\in E(T)} w_{x,y}(e) \leq 4d$.
\end{lemma}
We emphasise that this sum is taken over all ordered pairs of adjacent vertices. The proof of \Tr{tree linear} is based on the fact that, for a fixed $e$, the values $w_{x,y}(e)$ decay fast with the distance $d(xy,e)$, and even more so in the direction of $r$. (Here, we define the distance between two edges $xy,wz\in E$ to be $d(xy,wz):=\min\{d(x,w),d(y,z),d(x,z),d(y,w) \}$.) The following two propositions will yield quantitative bounds on the speed of this decay.

\begin{proposition}\label{first choice}
	Let $G$ be any graph, $x\in V(G)$, and $v\in N(x)$.
	Consider a CRW strategy that when at $x$ always chooses $v$ when that choice is available, otherwise it chooses each of the available options independently with probability $1/2$. Then for every $w\neq v \in N(x)$, the transition probabilities satisfy $q_{x,w}/q_{x,v}<1/2$.
\end{proposition}
\begin{proof}
	Let $d:=d(x)$.
	We have $q_{x,v}=1-\left(\frac{d-1}{d}\right)^2 = \frac{d^2 -(d^2-2d+1)}{d^2} = \frac{2d-1}{d^2}$ and $q_{x,w} = \left(\frac{d-1}{d}\right)^2 \frac1{d-1}$ since each $w\neq v$ is chosen with equal probability, and only when $v$ is not among the options. Thus $q_{x,w}/q_{x,v} = \left(\frac{d-1}{d}\right)^2  \frac{d^2}{(d-1) (2d-1)} = \frac{d-1}{2d-1} < 1/2$ as claimed.
\end{proof}

\begin{proposition}\label{last choice}
	Let $G$ be any graph, $x\in V(G)$, and $v,z \in N(x)$ where $v\neq z$.
	Consider a CRW strategy that when at $x$ always chooses $v$ when that choice is available, never chooses $z$ unless there is no other option, and it chooses each of the other available options independently with probability $1/2$. Then the transition probabilities satisfy $q_{x,z}/q_{x,v}=\frac1{2d(x)-1}$. 
\end{proposition}
\begin{proof}
	As in the proof of Proposition~\ref{first choice}, we have $q_{x,v}=\frac{2d(x)-1}{d(x)^2}$. Easily, $q_{x,z} =  \frac1{d(x)^2}$. Thus $q_{x,z}/q_{x,v} = \frac{d(x)^2}{(2d(x)-1)d(x)^2}= \frac1{2d(x)-1}$ as claimed.
\end{proof}

Armed with these propositions we can now prove \Lr{edge bound}.
\begin{proof}[Proof of \Lr{edge bound}]
	Fix $e\in E(T)$. We split $\Sigma:= \sum_{x,y\in V(T), xy\in E(T)} w_{x,y}(e)$ as a sum $\Sigma= \sum_{i\in \N} \Sigma_i$ of `layers' $\Sigma_i$, corresponding roughly to distance from $e$, and show that $\Sigma_i$ decays exponentially in $i$. 
	
	Let $P$ be the path from $e$ to $r$ in $T$ (excluding $e$), and let $L_0$ be the set of all edges of $P$ and all  edges incident with $P$ (including $e$). Let $\Sigma_0:= \sum_{x,y\in V(T), xy\in L_0} w_{x,y}(e)$ be the total weight assigned to $e$ by pairs of adjacent vertices of $L_0$. Define $L_i, i\geq 1$ recursively as the set of edges incident with $L_{i-1}$ not contained in $\bigcup_{j<i} L_j$, and let $\Sigma_i:= \sum_{x,y\in V(T), xy\in L_i} w_{x,y}(e)$.
	
	\begin{claim}
		The following inequalities hold\begin{enumerate}[(i)]
			\item \label{sum 0}
			$\Sigma_0 \leq 2d$,
			\item \label{sum i}
			$\Sigma_i \leq \Sigma_{i-1}/2.$
		\end{enumerate}
	\end{claim}

	\begin{poci} Let $x_1, x_2, \ldots, x_k$, where $x_k=r$ be the vertices of $P$ as they appear from $e$ to $r$. Recall that $w_{x_{i+1},x_i}(e)=0$ for every $i\geq 1$, as $e$ is separated from $x_{i+1}$ by removing $x_{i}$ and thus does not contribute to the sum in the formula for $H_{x_{i+1}}(x_{i})$ from Lemma \ref{essentialedge}. 
		In the other direction, we claim $w_{x_i,x_{i+1}}(e) < (1/2)^i$ for every $i\geq 1$. Indeed, by \Prr{first choice} we have $\frac{w_{x_i,x_{i+1}}(x_{i-1},x_i)}{w_{x_i,x_{i+1}}(x_i,x_{i+1})}<1/2$ because this ratio coincides with the ratio of the corresponding transition probabilities $\frac{q_{x_i, x_{i-1}}}{q_{x_i, x_{i+1}}}$ by the definitions. Moreover, at each $x_j, 1\leq j<i$, the strategy $\sigma^{x_i x_{i+1}}$ makes the same choices as $\sigma^{x_j x_{j+1}}$, hence \[\frac{w_{x_i,x_{i+1}}(x_{j-1},x_j)}{w_{x_i,x_{i+1}}(x_j,x_{j+1})} = \frac{w_{x_j,x_{j+1}}(x_{j-1},x_j)}{w_{x_j,x_{j+1}}(x_j,x_{j+1})}<1/2\] by \Prr{first choice} again, with the convention that $x_0 x_1 =e$. Our claim follows by multiplying these fractions for $j$ ranging from 1 to $i$. 
		
		For each of the at most $d-1$ other edges $x_i z\neq e$ of $L_0$ incident with $x_i$, where $i\geq 1$, we use the rather generous bound $w_{x_i,z}(e) < (1/2)^{i-1}$, which is true by similar arguments because $\frac{w_{x_i,z}(x_{i-1},x_i)}{w_{x_i,z}(x_i,z)}< 1$. Again we have $w_{z,x_i}(e)=0$.
		
		Finally, we have $w_{x_0,x_1}(e) = w_{x_1,x_0}(e) = 1$ since $e=x_0 x_1$. Adding these contributions we obtain $\Sigma_0 \leq 2+ \sum_{i\geq 0}(d-1) (1/2)^i= 2+2(d-1) =2d$ as claimed.\end{poci}
	
	\begin{pocii}
		Let $yv\in L_{i-1}$. We will bound the contribution of the edges incident with $yv$ to $\Sigma_i$ in terms of the contribution of $yv$ to $\Sigma_{i-1}$. For this, let $vw\in L_i$. Note that $v$ separates $w$ from $r$ and so firstly, $w_{w,v}(e)=0$. Secondly this implies $\sigma^{vw}$ avoids moving from $v$ to $y$ whenever possible by \eqref{sigma xy}. Thus \Prr{last choice} yields $\frac{w_{v,w}(yv)}{w_{v,w}(vw)}=\frac1{2d(v)-1}$. Moreover, when at a vertex other than $v$, the strategies $\sigma^{vw}$ and $\sigma^{yv}$ make the same choices since the directions of $r$ as well as of the corresponding target vertex coincide. Therefore, $\frac{w_{y,v}(f)}{w_{y,v}(g)}=\frac{w_{v,w}(f)}{w_{v,w}(g)}$ for every two edges $f,g$ incident with a common vertex on the $y$--$e$~path. It follows that $w_{v,w}(e) = \frac{w_{y,v}(e)}{2d(v)-1}$, and summing over all such neighbours $w$ of $v$ we obtain $\sum_{vw\in L_i} w_{v,w}(e) \leq \frac{w_{y,v}(e) (d(v)-1)}{2(d(v)-1)+1} < w_{y,v}(e)/2$. Applying this to each edge $yv\in L_{i-1}$, and adding together, noting that at most one end vertex $v$ of $yv$ is incident with edges in $L_i$ by construction, we finally deduce
		\[\Sigma_i = \sum_{yv\in L_{i-1}} \sum_{vw\in L_i} w_{v,w}(e) < \sum_{yv\in L_{i-1}} w_{y,v}(e)/2 = \Sigma_{i-1}/2,\]
		as desired.\end{pocii} Combining both parts of the Claim proves our statement, as
	$\Sigma = \sum_i \Sigma_i \leq 2 \Sigma_0\leq 4d$.
\end{proof}

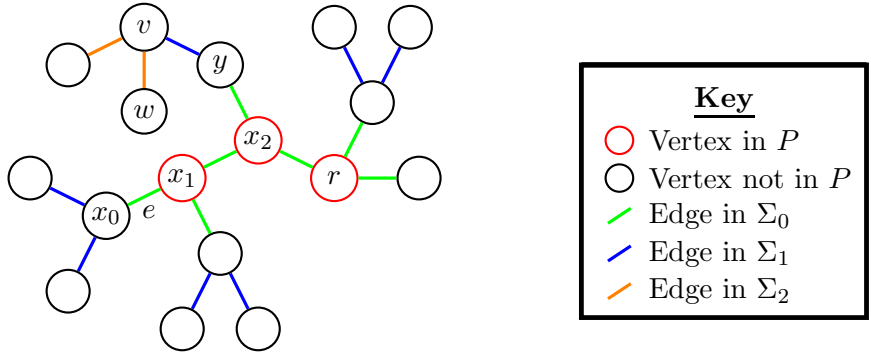
\begin{figure}
	\begin{center}
		\begin{tikzpicture}[xscale=.5,yscale=.5,knoten/.style={thick,circle,draw=black,fill=white},redknoten/.style={thick,circle,draw=red,fill=white},edge/.style={very thick,black},greenedge/.style={very thick,green},blueedge/.style={very thick,blue},orangeedge/.style={very thick,orange},dotedge/.style={very thick,dotted,black}]
		
		\begin{scope}[rotate=90]
		\node[knoten] (0) at (0,0) {$\phantom{\!\,r\!}$};
		\node[knoten] (1) at (-1,-2) {$\!\!x_0\!\!$};
		\node[knoten] (i) at (-3,-1) {$\phantom{\!\,r\!}$};
		\node[redknoten] (2) at (0,-4) {$\!\!x_1\!\!$};
		\node[knoten] (k) at (-2,-5) {$\phantom{\!\,r\!}$};
		\node[knoten] (g) at (-4,-4) {$\phantom{\!\,r\!}$};
		\node[knoten] (h) at (-4,-6) {$\phantom{\!\,r\!}$};
		\node[redknoten] (3) at (1,-6) {$\!\!x_2\!\!$};
		\node[knoten] (x) at (3,-5) {$ \!y\! $};
		\node[knoten] (y) at (4,-3) {$  v $};
		\node[knoten] (n) at (1.77,-3) {$ \!w\! $};
		\node[knoten] (z) at (3,-1) {$\phantom{\!\,r\!}$};
		\node[redknoten] (r) at (0,-8) {$ \!\,r\,\!$};
		\node[knoten] (a) at (0,-10.23) {$\phantom{\!\,r\!}$};
		\node[knoten] (b) at (2,-9) {$\phantom{\!\,r\!}$};
		\node[knoten] (c) at (4,-10) {$\phantom{\!\,r\!}$};
		\node[knoten] (d) at (4,-8) {$\phantom{\!\,r\!}$};
		
		\draw[blueedge] (0) to (1);
		\draw[greenedge] (1) to node[pos=0.65,color=black,below]{$e$} (2);
		\draw[blueedge] (1) to (i);
		\draw[greenedge] (2) to (3);
		\draw[greenedge] (2) to (k);
		\draw[blueedge] (k) to (g);
		\draw[blueedge] (k) to (h);
		\draw[blueedge] (x) to (y);
		\draw[orangeedge] (y) to (z);
		\draw[orangeedge] (y) to (n);
		\draw[greenedge] (3) to (x);
		\draw[greenedge] (3) to (r);
		\draw[greenedge] (r) to (a);
		\draw[greenedge] (r) to (b);
		\draw[blueedge] (b) to (c);
		\draw[blueedge] (b) to (d);
		\end{scope}
		%%%%%%Key%%%%%%%
		\begin{scope}[shift={(16,1)}]
		\draw[line width=1.8pt] (-1.5,2) -- (6,2)  -- (6,-4.7) -- (-1.5,-4.7) -- (-1.5,2) -- (6,2);
		
		\draw (1.2,1) node[anchor=west]{\underline{\textbf{Key}}};
		
		\node[redknoten] (1) at (-.5,0) {};
		\draw (0,0) node[anchor=west]{Vertex in $P$};
		
		\node[knoten] (1) at (-.5,-1) {};
		\draw (0,-1) node[anchor=west]{Vertex not in $P$};
		
		\draw[greenedge] (-.8,-2.2)--(-.2,-1.8);
		\draw (0,-2) node[anchor=west]{Edge in $\Sigma_0$};
		
		\draw[blueedge] (-.8,-3.2)--(-.2,-2.8);
		\draw (0,-3) node[anchor=west]{Edge in $\Sigma_1$};

		\draw[orangeedge] (-.8,-4.2)--(-.2,-3.8);
		\draw (0,-4) node[anchor=west]{Edge in $\Sigma_2$};
		
		\end{scope}
		
		\end{tikzpicture}\end{center}
	\caption{Illustration of the edge sets $\Sigma_{i}$ and path $P$ in the proof of \Lr{edge bound}. The vertices $y,v$ and $w$ are labelled consistently with the proof of \Lr{edge bound}, Claim \eqref{sum i} in the case where $yv \in L_1$ and $vw\in L_2$.}\label{edgeweights}
\end{figure}

It is now easy to complete the proof of \Tr{tree linear}.

\begin{proof}[Proof of \Tr{tree linear}]
	By \eqref{essentialedge} we have
	\[\sum_{x,y\in V(T), xy\in E(T)} \Htwo xy < \sum_{x,y\in V(T), xy\in E(T)} 2 \sum_{e\in E(T)} w_{x,y}(e).\] Changing the order of summation, and then applying \Lr{edge bound} to each summand we bound the right hand side by
	\[2 \sum_{e\in E(T)} \left(\sum_{x,y\in V(T), xy\in E(T)}  w_{x,y}(e) \right) \leq 2 \sum_{e\in E(T)} 4d = 2(\abs{V(T)}-1) 4d \leq 8 d\abs{V(T)}.\qedhere\]
\end{proof}

\section{Infinite Graphs and Cover Time of Tori}
In this section we bound the cover time of the $d$-dimensional discrete torus $\mathbb{Z}_k^d$, which has $n=k^d$ vertices. Here we think of the dimension $d$ as being fixed while the side length $k$ grows. In order to prove a linear bound on the cover time, we will instead consider the infinite limit $\mathbb{Z}^d$, and infinite (but locally finite) graphs more generally.

For infinite graphs it is meaningless to ask about the CRW cover time, but still interesting to ask about hitting times. The most fundamental question is whether these can be made finite, which corresponds to asking for positive recurrence.
\begin{definition}
A graph is \emph{positive choice recurrent (PCR)} if there exists an unchanging strategy for the CRW such that the expected return time to any given vertex is finite. A graph is \emph{strongly PCR} if for every $p\in(0,1)$ there exists an unchanging CRW strategy such that expected return times are finite for the process which, at every time step, takes a step of the CRW with that strategy with probability $p$ and a step of the simple random walk otherwise.\end{definition}

A natural question is whether there is a strategy under which the walk becomes a transient Markov chain. The answer is always yes: fixing a root $r$ and giving each edge $uv$ weight $2^{\min(d(u,r),d(v,r))}$ produces a suitable weighting to apply Corollary~\ref{double-weight}. This weighted graph is transient because any infinite geodesic starting at the root has total resistance $2$ (see e.g.~\cite[Theorem 2.3]{LPbook}), and taking other edges into account cannot increase the effective resistance to infinity.

While positive recurrence is the property which will be useful to us, we might also ask for the weaker property of \textit{choice recurrence}, where we simply require return times to be almost surely finite. It is possible for a graph to be choice recurrent but not PCR; indeed, there are graphs which are recurrent under the SRW but not PCR. 
\begin{remark}Proposition \ref{binary} implies that any graph of maximum degree $3$ is PCR. This is not true for higher degrees, since for the infinite $4$-regular tree any strategy is more likely to move away from a given target vertex than towards it.\end{remark}

Note that $\mathbb Z^d=\mathbb Z^{d-1}\square\mathbb Z$, where $\square$ indicates the Cartesian product. We will need the following result about Cartesian products of PCR graphs.
\begin{lemma}\label{prod}If $G$ is PCR, $H$ is strongly PCR and both $G,H$ are regular, then $G\square H$ is PCR.\end{lemma}
\begin{proof}
Define the $p$-product of two time-homogeneous Markov chains $A,B$ to be the chain with state space $S(A)\times S(B)$ where at each time step with probability $p$ a transition of $B$ occurs, and otherwise a transition of $A$ occurs. If both chains are irreducible and positive recurrent, then so is the $p$-product (this follows easily from the existence of stationary distributions). Now we define a strategy for the choice random walk on $G\square H$ as follows. If at least one of the choices given is a move in the $H$ co-ordinate, we make such a move. Now the probability of exactly one of the options being a move in $H$ is $\frac{2rs}{(r+s)^2}$, where $G$ is $r$-regular and $H$ is $s$-regular, and the probability of both options being moves in $H$ is $\frac{s^2}{(r+s)^2}$. Thus, conditional on at least one option being in $H$, both are in $H$ with probability $\frac s{2r+s}$. There is a strategy on $H$, for this probability of having two choices, which reaches the root in finite time; whenever we move in the $H$ co-ordinate we use this strategy. If both choices are moves in $G$ then we follow the appropriate strategy for the random walk with two choices in $G$. The resulting Markov chain is the $\frac{2rs+s^2}{(r+s)^2}$-product of positive recurrent Markov chains on $G$ and $H$, hence positive recurrent.
\end{proof}
The same argument shows that if in addition $G$ is strongly PCR then so is $G\square H$. Lemma \ref{prod} allows us to conclude that $\mathbb Z^d$ is PCR and consequently obtain a bound on its cover times and hitting times.

\newtheorem*{thm:covergrid}{Theorem \ref{covergrid}}
\begin{thm:covergrid}
For any $d$, the cover time of the CRW in the finite $d$-dimensional torus $\mathbb Z_k^d$ or grid $[k]^d$ is $\BT{n}$ and the hitting time is $\BT{k} = \BO{n^{1/d}}$, where $n$ is the number of vertices.\end{thm:covergrid}
\begin{proof}
Note that $\mathbb Z$ is strongly PCR, since always moving towards $0$ if possible gives a random walk which moves towards $0$ with probability $\frac 12+\frac p4>\frac 12$ (where $p$ is the probability of taking a step of the CRW). Inductively applying Lemma~\ref{prod} implies $\mathbb Z^d$ is PCR for any $d$, and so the hitting time to a neighbour is some constant $c_d$. This gives an upper bound on the hitting time to a neighbour in $\mathbb Z_k^d$ or $[k]^d$, and the strategy of visiting the vertices of a Hamilton path in order gives a cover time of less than $c_dk^d$. Similarly the hitting time $\Htwo xy$ is bounded by constant times the length of the shortest path between $x$ and $y$, and the worst-case value of this is $d\lfloor k/2 \rfloor$ for the torus and $dk$ for the grid, so $\BT{k}$. Both bounds are trivially best possible up to the constant.\end{proof}

\section{Hitting and Cover Times in Expanders}\label{S:regular}
In this section we prove bounds on the cover and hitting times of the CRW on a graph $G$ in terms of fundamental parameters. First we introduce our notation. Let $G$ be a graph with $n$ vertices, and write $\dmax$, $\dmin$ and $\davg$ for the maximum, minimum and average degree of $G$ respectively. Let $\trel$ be the \emph{relaxation time} of $G$, defined as $\frac{1}{1-\lambda_2}$, where $\lambda_2$ is the second largest eigenvalue of the transition matrix of the lazy random walk (LRW) on $G$ with loop probability $1/2$. Recall that $\thit$ is maximum over all pairs $uv\in V$ of the expected time it takes the SRW to reach $u$ from $v$. Our first result bounds the CRW cover time.
\begin{theorem}\label{trelbdd}
	For any connected $n$-vertex graph $G$ the following holds 
	\[\ttwo(G) = \BO{\thit\cdot \log(\dmax)\cdot \log\left( \frac{\davg\cdot \trel \cdot \log  n}{\dmin} \right) },\]
\end{theorem}
We also bound hitting times. First we define the exponent $\gamma_d=\log_d\frac{d^2}{2d-1}$; note that $\gamma_d$ is increasing in $d$, $\gamma_d<1$ and $1-\gamma_d\sim1/\log_2 d$. Also recall that for a set $S\subseteq V$ let $\pi(S)=\sum_{s\in S} \pi(s)$ be the stationary probability of $S$. 
\begin{theorem}\label{trelhit}For any graph $G$, and any $x\in V $ and $S\subset V$ we have
	\[\Htwo{x}{S}\leq 12 \cdot \pi(S)^{-\gamma_{\dmax}}\cdot \trel\cdot \ln n;\]
	this bound also holds for return times. Consequently,
	\[\ttwoh(G)\leq 12\left(\frac{n\cdot \davg}{ d_{\mathsf{min}}}\right)^{\gamma_{\dmax}}\cdot \trel\cdot \ln n.\]
\end{theorem}

We say that a sequence of graphs $(G_n)$ is a \emph{sequence of expanders} if $\trel(G_n) = \BT{1}$. Theorems \ref{trelbdd} \& \ref{trelhit} yield the following corollary:

\newtheorem*{thm:exp}{Theorem \ref{trelbddcor}}
\begin{thm:exp}
For every sequence $(G_n)_{n\in \N}$ of bounded degree $n$-vertex expanders, we have\[\ttwo(G_n)=\BO{n \log \log n}\qquad \text{and}\qquad \ttwoh(G_n)\leq n^{\alpha}\quad\text{for some fixed }\alpha<1 .\]
\end{thm:exp}  
These are significantly less than the corresponding cover and hitting times by the SRW, which are $\Theta(n \log n)$ and $\BT{n}$ respectively \cite[Thm.\ 10.1]{aldousfill}.

Theorems \ref{trelbdd} and \ref{trelhit} will follow from Theorem \ref{nonregboostnew} below. For a given graph $G$, we consider possible \emph{trajectories} of a (non-lazy) walker, that is, finite sequences of vertices in which any two consecutive vertices are adjacent; the length of a trajectory will be the number of steps taken. 
In the following we use bold characters to denote trajectories in $G$ and if $u\in V(G)$ then $u$ will denote the length $0$ trajectory from $u$. Fix a non-negative integer $t$ and a set $S$ of trajectories of length $t$. Let $p_{\mathbf{x},S}$ denote the probability that extending a trajectory $\mathbf{x}$ to length $t$ according to the law of a SRW results in a member of $S$. Let $q_{\mathbf{x},S}$ denote the corresponding probability under the CRW law; this probability will depend on the particular strategy used. This function can encode probabilities of many events of interest such as ``the graph is covered by time $t$'', ``the walk is in a set $X$ at time $t$'' or ``the walk has hit a vertex $x$ by time $t$'' for example. However, let us emphasise that our result in fact applies to \emph{any} possible event. 
\begin{theorem}\label{nonregboostnew}Let $G$ be a graph, $u\in V$, $t > 0$ and $S$ be a set of trajectories of length $t$ from $u$.  Then there exists a strategy for the CRW such that $q_{u,S} \geq \left( p_{u,S} \right)^{\gamma_{\dmax}}$.
\end{theorem}
We also give an analogue of Theorem \ref{nonregboostnew} for bad events. This analogue, unlike Theorem \ref{nonregboostnew}, gives an exponent which does not depend on the maximum degree $\dmax$ of $G$, and so a significant reduction is possible even if $\dmax$ is large.
\begin{theorem}\label{antiboost}Let $G$ be a graph, $u\in V$, $t > 0$, and $S$ be a set of trajectories of length $t$ from $u$. Then there exists a strategy for the CRW such that $q_{u,S} \leq \left( p_{u,S} \right)^2$.
\end{theorem}
\begin{remark}The exponent $2$ in Theorem \ref{antiboost} is best possible, since we have equality whenever $t=1$ and therefore also when $t>1$ but every trajectory of the SRW of length $t-1$ has the same probability to reach $S$. Similarly the exponent given in Theorem \ref{nonregboostnew} is best possible, as evidenced by the case where this probability is $1/\dmax$ for every trajectory of length $t-1$.\end{remark}

After stating two technical lemmas in Section \ref{game}, we then explain an alternative way of considering the CRW in Section \ref{gadget}, which enables the proof of Theorems \ref{nonregboostnew} and \ref{antiboost} to be completed. To motivate the importance of Theorem \ref{nonregboostnew} we shall begin by showing how it implies our main results on cover time and hitting times.  
\subsection{Deducing Theorems \ref{trelbdd} \& \ref{trelhit} from Theorem \ref{nonregboostnew}}\label{S:deduce}
In order to prove our main bounds from the key tool, Theorem \ref{nonregboostnew}, we must first overcome the obstacle that Theorem \ref{nonregboostnew} is expressed in terms of the SRW probabilities, whereas our bounds involve the relaxation time, which is defined in terms of the LRW. The reason for using the LRW to define relaxation time is to ensure that the associated Markov chain is aperiodic. Our next lemma resolves this issue by relating the relaxation time to SRW probabilities.

Write ${p}_{x,\cdot }^{(t)}$ and $\tilde{p}_{x,\cdot }^{(t)}$ for the distribution of the SRW and LRW respectively after $t$ steps started at $x$, and write $\pi(S)$ for the stationary probability of a set $S$ (note that the two walks have the same stationary distribution).

\begin{lemma}\label{lazyconv} For any graph $G$, $S\subset V $ and $ x \in V$ there exists $t\leq 4\trel\ln n$ such that \[p_{x,S}^{(t)} \geq \pi(S)/3 .\]  
\end{lemma}

\begin{proof}
	If $G$ is bipartite, then we may find a subset $\bar S\subseteq S$ which lies entirely within one part satisfying $\pi(\bar{S})\geq \pi(S)/2$. Otherwise the SRW is aperiodic and we set $\bar S=S$. We now consider the multigraph $\bar{G}$ formed from $G$ by contracting $\bar{S}$ to a single vertex $\bar{s}$, retaining all edges (with edges inside $\bar{S}$ becoming loops at $\bar{s}$). Retaining edges ensures that the stationary probability of $\bar{s}$ in $\bar{G}$ is precisely $\pi(\bar{S})$. Let $\bar{\lambda}_2$ be the second largest eigenvalue of the LRW on $\bar{G}$. Then for any $x\notin \bar{S}$ and $t\geq 0$, by \cite[(12.11)]{levin2009markov}, we have $\abs{\tilde{p}_{x,\bar{s}}^{(t)}-\pi(\bar{S})}\leq \sqrt{ \pi(\bar{S})/\pi(x)}\cdot \e^{-t(1-\bar{\lambda}_2)}$. It follows that if we run the LRW on $\bar{G}$ for $T=\log(3/\sqrt{\pi(\bar{S})\pi(x)})/(1-\bar{\lambda}_2)$ steps then \[\tilde{p}_{x,\bar{s}}^{(T)} \geq\pi(\bar{S}) - \sqrt{\frac{\pi(\bar{S})}{\pi(x)}} \cdot\frac{\sqrt{\pi\left(\bar{S}\right)\pi(x)}}{3} \geq   \frac{2\pi(\bar{S})}{3}.\] Now, we can express the density of the LRW by $\tilde{p}_{x,S}^{(T)} = \Ex{{p}_{x,S}^{(X_T)}}$, where the random variable $X_T\sim\operatorname{Bin}(T,1/2)$ is the number of non-lazy steps taken by the LRW in time $T$. Thus \[\max_{t\leq T}p_{x,S}^{(t)} \geq \tilde{p}_{x,S}^{(T)}\geq  \frac{2\pi(\bar{S})}{3} \geq  \frac{\pi(S)}{3}.\]We can assume $n\geq2$ or else the result holds trivially, so $\log(3/\sqrt{\pi(\bar{S})\pi(x)})\leq \log 3 + 2\log n\leq 4\log n$. Finally, \cite[Cor.\ 3.27]{aldousfill} gives that $\bar{\lambda}_2 \leq \lambda_2$, so $T\leq4 \trel\ln n $.\end{proof}

Our strategy to bound the cover time will be to emulate the SRW until most of the vertices are covered, only using the additional strength of the CRW when there are few uncovered vertices remaining. We will need a simple lemma to bound how long the first stage takes.
\begin{lemma}\label{vacant} 
	Let $U(t)$ be the number of unvisited vertices at time $t$ by a SRW on a graph and let $T_{n/2^x} $ be the number of SRW steps taken before $U\leq n/2^x  $. Then 
	\[ \Ex{U(2x\cdot \thit)} \leq  \frac{n}{2^x} \qquad \text{and} \qquad \Ex{T_{n/2^x}}\leq 4(x+1)\thit.  \]    
\end{lemma}
\begin{proof}Let $v \in V$. Then by Markov's inequality $\Pru{w}{X_t \neq v, \; \forall  0\leq t \leq 2\thit}\leq 1/2$, for any $w \in V$. Thus the probability $v$ is not visited by time $ 2x\cdot\thit$ is at most $2^{-x}$ by sub-multiplicity and so the expected number of unvisited vertices at time $ 2x\cdot\thit $ is at most $n\cdot 2^{-x} $. 
	
	By the above $\Ex{U(2(x+1)\thit)}\leq n/(2\cdot2^x) $ and so $\Pr{U(2(x+1)\thit)\geq n/2^x }\leq 1/2 $ by Markov's inequality. Considering sections of length $2(x+1)\thit $ separately, and continuing until one section covers the required number of vertices, we use in expectation at most two such sections, thus $\Ex{T_{n/2^x}}\leq 4(x+1)\thit$.
\end{proof} 
We now have what we need to prove the cover and hitting time bounds.

\begin{proof}[Proof of Theorem \ref{trelbdd}]For convenience we write $\gamma=\gamma_{\dmax}$. 
	We first emulate the SRW (i.e.\ set $\axyz xyz=1/2$ for all $x,y,z\in V(G)$ with $y,z\in\Gamma(x)$) until all but $m=\left\lfloor n/\log^C n\right\rfloor$ 
	vertices have been visited, for some $C$ to be specified later. Let $\tau_1$ be the expected time to complete this phase. Then, by Lemma \ref{vacant}, we have  $\tau_1 \leq 4\thit\cdot C\log_2\log n$. 
	
	We cover the remaining vertices in $m$ different phases, labelled $m,m-1,\ldots,1$, each of which reduces the number of uncovered vertices by $1$. In phase $i$, a set of $i$ vertices are still uncovered, and we write $S_i$ for this set. By Lemma \ref{lazyconv} for any vertex $x$ there is some $t\leq 4\trel\log n $ such that \[p_{x,S_i}^{(t)} \geq \frac{\pi(S_i)}{3} = \frac{1}{3}\cdot \frac{\sum_{s\in S_i} d(s)}{n\davg}\geq\frac{ d_{\mathsf{min}}\cdot i}{3n\davg},\] and thus 
	$q_{u,S_i}^{(t)} \geq \left(\dmin\cdot i/(3n\davg ) \right)^{\gamma}$ by Theorem \ref{nonregboostnew}. Since from any starting point we can achieve this probability of hitting a vertex in $S_i$ within the next $4\trel\log n$ steps, the expected number of attempts needed to achieve this is at most $\left(\dmin\cdot i/(3n\davg)\right)^{-\gamma}$, meaning that the expected time required to complete phase $i$ is at most
	\[\BO{ \left(\frac{n\cdot \davg}{i\cdot \dmin}\right)^{\gamma}\cdot \trel\cdot \log n}.\]
	Hence the expected time $\tau_2$ to complete all $m$ phases satisfies
	\begin{align*}
	\tau_2&= \sum_{i=1}^{n/\log^C n} \BO{ \left(\frac{n \davg}{i \dmin}\right)^{\gamma}\trel\log n}\\
	&= \BO{ \left(\frac{n \davg}{d_{\mathsf{min}}}\right)^{\gamma} \trel\log n}\sum_{i=1}^{n/\log^C n}i^{-\gamma}. \end{align*}
	Then, since $\sum_{i=1}^{n/\log^C n}i^{-\gamma}\leq \left(n/\log^C n\right)^{1-\gamma}\cdot  \sum_{i=1}^{n/\log^C n}i^{-1} \leq    \left(n/\log^C n\right)^{1-\gamma} \cdot \log n  $, we have
	\begin{align}
	\tau_2&=\BO{\left(\frac{n\davg}{\dmin}\right)^{\gamma} \trel\log n} \cdot \BO{\left(\frac{n}{\log^C n}\right)^{1-\gamma} \cdot \log n  }\notag \\
	&= \BO{   n \cdot \left(\frac{\davg}{\dmin}\right)^{\gamma}\cdot \trel \cdot \frac{\log^2 n}{\log^{C(1-\gamma)} n} }.\label{t2bdd}
	\end{align}
	For the first bound we choose $C=  \log\left( (\frac{\davg}{\dmin})\cdot \trel \cdot \log^2 n \right)/ \left((1-\gamma)\cdot \log\log n\right)$ then since $\log^{C(1-\gamma)}n = (\davg/\dmin)\trel\cdot \log^2 n $ and $\gamma<1$ this gives $\tau_2  =\BO{n}$ by \eqref{t2bdd} above. Since in any graph $\thit=\Omega(n) $,\footnote{Let $\tau_v$ be the first time that $v$ is visited during a random walk from $u$. Then $\sum_{v\neq u}\tau_v \geq \sum_{i=1}^{n-1}i =\binom{n}{2}$, since each $ \tau_v$ is distinct. Thus $t_{\mathsf{hit}}\geq \sum_{v\neq u} H_u(v) /n \geq (n-1)/2$.} the total time is therefore $\BO{\tau_1}$, and for this value of $C$ we have 
	\[\tau_1 = \BO{\frac{\log\left( (\frac{\davg}{\dmin})\cdot \trel \cdot \log^2 n \right)}{(1-\gamma)\cdot \log\log n} \thit\log\log n} = \BO{\thit\cdot \log(\dmax)\cdot \log\left( \frac{\davg\cdot \trel \cdot \log  n}{\dmin} \right) } ;\]
	since $1-\gamma=\Theta(1/\log \dmax )$. \end{proof}

\begin{proof}[Proof of Theorem \ref{trelhit}]Write $T= 4\cdot \trel\cdot \ln n$. For any $x\in V$ and $S\subset V$, Lemma \ref{lazyconv} gives a $t\leq T$ such that $p_{x,S}^{(t)}\geq  \pi(S)/3$, and Theorem \ref{nonregboostnew} consequently gives a strategy for the CRW such that $q_{x,y}^{(t)} \geq (\pi(y)/3)^{\gamma}$. Thus for any target set $S$ and start vertex $x$ we need in expectation at most $\left( 3/\pi(S)\right)^{\gamma} $ attempts to hit $S$ in at most $T$ steps, since if an attempt fails, ending at some vertex $z$, we have the same bound on the probability of hitting $S$ from $z$. Therefore there is a strategy for the CRW where the hitting time $H_{x}(S)\leq 12\cdot \pi(S)^{-\gamma}\cdot \trel \log n$. The second result follows since for any vertex $\pi(v)\geq \frac{\dmin}{n\davg}$.\end{proof}

\subsection{The Max Choice and Min Choice Operations}\label{game}
In this section we introduce two operators which represent the effect of making optimal choices for a single step of the random walk, assuming that the effects of choice on future steps are already known, and prove inequalities relating them to power means.

Define the max choice operator $\operatorname{MC}_2\colon [0,\infty)^m\to[0,\infty)$ as follows:
\begin{equation}\label{mcopp}\operatorname{MC}_2(x_1,\dots,x_d) = \frac{1}{m^2}\sum_{i=1}^d\sum_{j=1}^d \max\{x_i,x_j\}.\end{equation}

For $p\in \mathbb R\setminus\{0\}$, the $p$-power mean $M_p$ of non-negative reals $x_1,\ldots,x_m$ is defined by 
\[M_p(x_1,\ldots,x_m)=\left(\frac{x_1^p+\cdots+x_m^p}{m}\right)^{1/p}.\] We use a key lemma which could be be described as a multivariate anti-convexity inequality. 
\begin{lemma}\label{anticonv}For any $1\leq d\leq m$ and $x_1\ldots x_d \in [0,1]$ we have 
	\[M_{\gamma^{-1}_m}(x_1,\ldots,x_d)\leq \operatorname{MC}_2(x_1,\dots,x_d).\]
\end{lemma}
\begin{proof}
By the power-mean inequality, since $\gamma_m^{-1}\leq\gamma_d^{-1}$ it is sufficient to prove the case $m=d$. We show this by induction on $d$; we have equality for $d=1$. Suppose that either $d=2$ or $d\geq 3$ and the result holds for $d-1$.
Without loss of generality, using symmetry and homogeneity of both operators, we may assume that $\max\{x_1,\ldots,x_d\}=x_d=1$. 

We first claim that we may further assume $x_1=\cdots=x_{d-1}$. If $d=2$ this claim is trivial. If $d\geq 3$ then write $\bar x=M_{\gamma_d^{-1}}(x_1,\ldots,x_{d-1})$. 
Note that
\[M_{\gamma_d^{-1}}(x_1,\ldots,x_{d-1},x_d)=M_{\gamma_d^{-1}}(\bar x,\ldots,\bar x,x_d).\]
Also we have
\begin{align*}\operatorname{MC}_2(x_1,\ldots,x_{d-1},x_d)&=\frac{2d-1}{d^2}x_d+\bfrac{d-1}{d}^2\operatorname{MC}_2(x_1,\ldots,x_{d-1})\\
&\geq\frac{2d-1}{d^2}x_d+\bfrac{d-1}{d}^2M_{\gamma_{d-1}^{-1}}(x_1,\ldots,x_{d-1})\\
&\geq\frac{2d-1}{d^2}x_d+\bfrac{d-1}{d}^2\bar x\\
&=\operatorname{MC}_2(\bar x,\ldots,\bar x,x_d),
\end{align*}
where the first inequality uses the assumption that the result holds for $n-1$ and the second uses the power-mean inequality. Thus replacing $x_1,\ldots,x_{d-1}$ by $\bar x,\ldots,\bar x$ does not increase the difference between the two operators, proving the claim.

Next we claim that the function $f(x)=M_{\gamma_d^{-1}}(x,\ldots,x,1)$ is convex. Since $\operatorname{MC}_2(x,\ldots,x,1)$ is linear, and the two functions agree at $0$ (by choice of $\gamma_d$) and at $1$, this will complete the proof. Indeed, we have $f(x)^{\gamma_d^{-1}}=\frac{d-1}{d}x^{\gamma_d^{-1}}+\frac1d$, giving $f'(x)=\frac{d-1}{d}\bfrac{x}{f(x)}^{\gamma_d^{-1}-1}$. Also, $\frac{x^{\gamma_d^{-1}}}{f(x)^{\gamma_d^{-1}}}$ is an increasing function of $x$; since $\gamma_d^{-1}-1>0$, we have $f'(x)$ is increasing, as required.
\end{proof}
Lemma \ref{anticonv} will be used to prove Theorem \ref{nonregboostnew}. In order to prove Theorem \ref{antiboost} we will need a corresponding inequality for an appropriate operator. To that end we define the \textit{min choice} operator $\operatorname{mC}_{2}:[0,\infty)^m\to [0,\infty)$ by \[\operatorname{mC}_{2}\left(x_1,\dots , x_m\right)  = \frac 1{m^2}\sum_{i=1}^m\sum_{j=1}^m\min\{x_i,x_j\}.\]
\begin{lemma}\label{newanticonv}
	For any $m\geq 1$ and non-negative reals $x_1,\ldots,x_m$ we have \[\operatorname{mC}_{2}(x_1,\ldots,x_m)\leq M_{1/2}(x_1,\ldots,x_m).\]
\end{lemma}
\begin{proof}
	Observe that
	\begin{align*}\biggl(\frac 1m\sum_{i=1}^m\sqrt{x_i}\biggr)^2&=\frac1{m^2}\sum_{i=1}^m\sum_{j=1}^m\sqrt{x_ix_j}\\
	&\geq\frac1{m^2}\sum_{i=1}^m\sum_{j=1}^m\min\{x_i,x_j\}\qedhere\end{align*}
\end{proof}
\subsection{The Tree Gadget for Graphs}\label{gadget}
In this section we prove Theorem \ref{nonregboostnew}. To achieve this we introduce the \textit{Tree Gadget} which encodes trajectories of length at most $t$ from $u$ in a rooted graph $(G,u)$ by vertices of an arborescence $(\mathcal{T}_t,\mathbf{r})$, i.e.\ a tree with all edges oriented away from the root $\mathbf{r}$. Given $(G,u)$ we represent each trajectory of length $i\leq t$ started from $u$ in $G$ as a node at distance $i$ from the root $\mathbf{r}$ in the tree $\mathcal{T}_t$. The root $\mathbf{r}$ represents the trajectory of length $0$ from $u$. There is an edge from $\mathbf{x}$ to $\mathbf{y}$ in $\mathcal{T}_t$ if $\mathbf{x}$ is obtained from $\mathbf{y}$ by deleting the final vertex.  

Also for $\mathbf{x} \in V(\ct_t)$ let $\Gamma^+(\mathbf{x}) = \{\mathbf{y} \in V(\ct_t) : \mathbf{x}\mathbf{y}\in E(\ct_t) \} $ be the offspring of $\mathbf{x}$ in $ \ct$; as usual we write $d^+(\mathbf{x})$ for the number of offspring. Write $\abs{\mathbf{x}}$ for the length of the trajectory $\mathbf{x}$. To prove Theorem \ref{nonregboostnew} we shall need to discuss SRW trajectories; let $W_u(k):=(X_i)_{i=0}^k$ be the trajectory of a SRW $X_i$ on $G$ up to time $k$, with $X_0=u$.

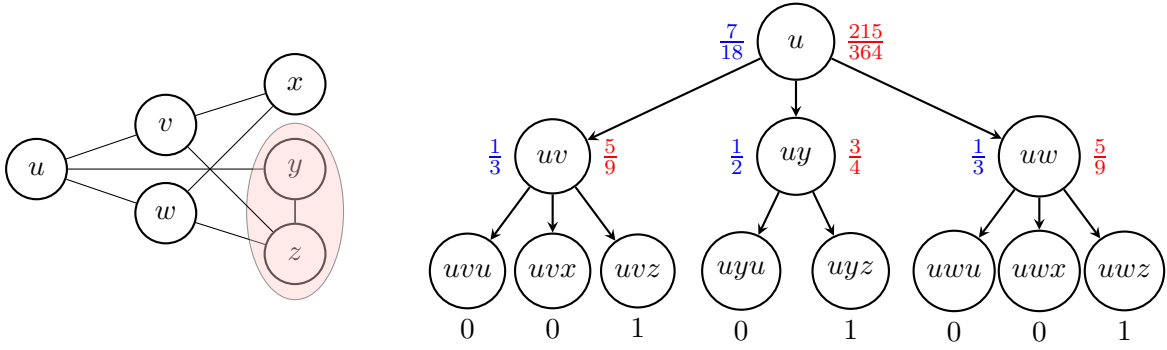
\begin{figure}
	\begin{subfigure}{.35\textwidth}
	\begin{tikzpicture}[xscale=.85,yscale=0.9,knoten/.style={thick,circle,draw=black,minimum size=.8cm,fill=white},wknoten/.style={thick,circle,draw=black,minimum size=.6cm,fill=white},edge/.style={black},dedge/.style={thick,black,-stealth}]
\node[knoten] (u) at (2,2) {$u$};
\node[knoten] (v) at (4,2.65) {$v$};
\node[knoten] (w) at (4,1.35) {$w$};
\node[knoten] (x) at (6,3.25) {$x$};
\node[knoten] (y) at (6,2) {$y$};
\node[knoten] (z) at (6,0.75) {$z$};
\draw[edge] (u) to (v);
\draw[edge] (u) to (y);
\draw[edge] (w) to (x);
\draw[edge] (w) to (z);
\draw[fill=red!20,opacity=0.4] (6,1.375) ellipse (0.75cm and 1.3cm);
\draw[edge] (v) to (x);
\draw[edge] (z) to (y);
\draw[edge] (u) to (w);
\draw[edge] (v) to (z);
\end{tikzpicture}
	\end{subfigure}%
\begin{subfigure}{.65\textwidth}
	\begin{tikzpicture}[xscale=0.8,yscale=0.76,knoten/.style={thick,circle,draw=black,minimum size=.6cm,fill=white},wknoten/.style={thick,circle,draw=black,minimum size=.6cm,fill=white},edge/.style={black},dedge/.style={thick,black,-stealth}]
		\node[knoten] (u1) at (2,6) [label=left:$\textcolor{blue}{\frac{7}{18}}$,label=right:$\textcolor{red}{\frac{215}{364}}$]{$\phantom{x}u\phantom{x}$};
		\node[knoten] (v2) at (-2,4)[label=left:$\textcolor{blue}{\frac{1}{3}}$,label=right:$\textcolor{red}{\frac{5}{9}}$] {$\;uv\;$};
		\node[knoten] (y2) at (2,4)[label=left:$\textcolor{blue}{\frac{1}{2}}$,label=right:$\textcolor{red}{\frac{3}{4}}$] {$\;uy\;$};
		\node[knoten] (w2) at (6,4)[label=left:$\textcolor{blue}{\frac{1}{3}}$,label=right:$\textcolor{red}{\frac{5}{9}}$] {$\;uw\;$};
		
		\node[wknoten] (u3) at (-3.4,2)[label=below:$0$] {$uvu$};
		\node[wknoten] (x3) at (-2,2)[label=below:$0$] {$uvx$};
		\node[wknoten] (z3) at (-0.6,2)[label=below:$1$] {$uvz$};
		
		\node[wknoten] (u31) at (1.1,2)[label=below:$0$] {$uyu$};
		\node[wknoten] (z31) at (2.9,2)[label=below:$1$] {$uyz$};
		
		\node[wknoten] (u32) at (4.6,2) [label=below:$0$]{$uwu$};
		\node[wknoten] (x32) at (6,2) [label=below:$0$]{$uwx$};
		\node[wknoten] (z32) at (7.4,2) [label=below:$1$]{$uwz$};

		\draw[dedge] (u1) to (v2);
		\draw[dedge] (u1) to (y2);
		\draw[dedge] (u1) to (w2);
		
		\draw[dedge] (v2) to (u3);
		\draw[dedge] (v2) to (x3);
		\draw[dedge] (v2) to (z3);
		
		\draw[dedge] (y2) to (u31);
		\draw[dedge] (y2) to (z31);
		
		\draw[dedge] (w2) to (u32);
		\draw[dedge] (w2) to (x32);
		\draw[dedge] (w2) to (z32);
	
		\end{tikzpicture}
			\end{subfigure}
	\caption{Illustration of a (non-lazy) walk on a non-regular graph starting from $u$ with the objective of being in $\{y,z\}$ at step $t=2$. The probabilities of achieving this are given in blue (left) for the SRW and in red (right) for the CRW.}
	
\end{figure}

\begin{proof}[Proof of Theorem \ref{nonregboostnew}]For ease of notation we write $\eta=1/\gamma_{\dmax}$. To each node $\mathbf{x}$ of the tree gadget $\ct_t$ we assign the value $q_{\mathbf{x},S} $ under the CRW strategy of preferring the choice which extends to a trajectory $\mathbf{y}\in \Gamma^+(\mathbf{x})$ giving a higher value of $q_{\mathbf{y},S}$. This is well defined because both the strategy and the values $q_{\mathbf{x},S}$ can be computed in a ``bottom up'' fashion starting at the leaves, where if $\mathbf{x} \in V(\ct_t)$ is a leaf then $q_{\mathbf{x},S} $ is $1$ if $\mathbf x\in S$ and $0$ otherwise. 
	
Suppose $\mathbf{x}$ is not a leaf. The controller is presented with two uniformly random offspring $\mathbf y,\mathbf z\in\Gamma^+(\mathbf x)$, and chooses $\mathbf y$ if $q_{\mathbf{y},S}\geq q_{\mathbf{z},S}$ and $\mathbf z$ otherwise. Thus we have
\begin{equation}\label{qqqq}q_{\mathbf{x},S}=
\frac1{d^+(\mathbf x)^2}\sum_{\mathbf y,\mathbf z\in\Gamma^+(\mathbf x)}\max\{q_{\mathbf{y},S},q_{\mathbf{z},S}\}
=\operatorname{MC}_{2}\left( \left(q_{\mathbf{y},S}\right)_{\mathbf{y}\in \Gamma^+(\mathbf{x})} \right).\end{equation} 

We define the following potential function $\Phi^{(i)}$ on the $i^{th}$ generation of the tree gadget $\ct$: 
\begin{equation}\label{Phi}\Phi^{(i)}= \sum\limits_{\abs{\mathbf{x}}=i}q_{\mathbf{x},S}^{\eta}\cdot \Pr{W_u(i) = \mathbf{x}};\end{equation}
where the sum ranges over all trajectories $\mathbf{x}$ of length $i$.
Notice that if $\mathbf{x}\mathbf{y}\in E(\ct_t)$ then \[\Pr{W_u(\abs{\mathbf{y}}) = \mathbf{y}} = \Pr{W_u(\abs{\mathbf{x}}) = \mathbf{x}}/d^+(\mathbf{x}).\] 
Also since each $\mathbf{y}$ with $abs{\mathbf{y}}=i$ has exactly one parent $\mathbf{x}$ with $\abs{\mathbf{x}}=i-1$ we can write 
\begin{equation}\label{PPhi}\Phi^{(i)} = \sum\limits_{\abs{\mathbf{x}}=i-1}\sum_{\mathbf{y} \in \Gamma^+(\mathbf{x})}q_{\mathbf{y},S}^{\eta}\cdot \frac{\Pr{W_u(i-1) = \mathbf{x}}}{d^+(\mathbf{x})  }.\end{equation} We now show that $\Phi^{(i)} $ is non-increasing in $i$. By combining \eqref{Phi} and \eqref{PPhi} we can see that the difference $\Phi^{(i-1)}-\Phi^{(i)}$ is given by  
\begin{align*}
&\sum\limits_{\abs{\mathbf{x}}=i-1} \left(q_{\mathbf{x},S}^{\eta}-\frac{1}{d^+(\mathbf{x})}\sum_{\mathbf{y} \in \Gamma^+(\mathbf{x})}q_{\mathbf{y},S}^{\eta} \right) \Pr{W_u(i-1) = \mathbf{x}}.
\end{align*}Recalling \eqref{qqqq}, to establish $\Phi^{(i-1)}-\Phi^{(i)}\geq 0$ it is sufficient to show the following inequality holds whenever $\mathbf{x}$ is not a leaf: 
\[ \operatorname{MC}_{2}\left(\left( q_{\mathbf{y},S}\right)_{\mathbf{y} \in \Gamma^+(\mathbf{x})} \right)^{\eta} \geq \frac{1}{d^+(\mathbf{x} ) }\sum_{\mathbf{y}\in \Gamma^+(\mathbf{x})}q_{\mathbf{y},S}^{\eta}.\]
Raising both sides to the power $1/\eta = \gamma_{\dmax}$, since $d^+(\mathbf x)\leq\dmax$ this inequality holds by Lemma \ref{anticonv}, and thus $\Phi^{(i)} $ is non-increasing in $i$. 
	
Observe $\Phi^{(0)} = q_{u,S}^{\eta}$. Also if $\abs{\mathbf{x}}=t$ then $q_{\mathbf{x},S}=1 $ if $\mathbf{x} \in S$ and $0$ otherwise. It follows that  
\[\Phi^{(t)} = \sum_{\abs{\mathbf{x}}=t}q_{\mathbf{x},S}^{\eta}\cdot \Pr{W_u(t) = \mathbf{x}}=\sum_{\abs{\mathbf{x}}=t}\mathbf{1}_{\{\mathbf{x}\in S\}}\cdot \Pr{W_u(t) = \mathbf{x}} = p_{u,S} .\] Thus since $\Phi^{(t)}$ is non-increasing $q_{u,S}^{\eta} = \Phi^{(0)}\geq \Phi^{(t)} = p_{u,S}  $, as required. \end{proof}

Theorem \ref{antiboost} now follows similarly to Theorem \ref{nonregboostnew}. 
\begin{proof}[Proof of Theorem \ref{antiboost}]
	Construct the tree gadget to height $t$. We associate each node $\mathbf x$ with the probability $q_{\mathbf x,S}$ under a strategy which always prefers the smaller value. For a leaf this is simply the indicator function $\mathbf{1}_{\{\mathbf x\in S\}}$, whereas for an internal vertex it is given by $\operatorname{mC}_2\bigl((q_{y,S})_{y\in\Gamma^+(\mathbf x)}\bigr)$.
	We define a potential function $\Psi$ by \[\Psi^{(i)}= \sum_{\abs{\mathbf{x}}=i}\Pr{W_u(i) = \mathbf{x}}\sqrt{q_{\mathbf{x},S}}.\]
	As before, 
	\[\Psi^{(t)}=\sum_{\abs{\mathbf{x}}=t}\Pr{W_u(t) = \mathbf{x}}\cdot \mathbf{1}_{\{\mathbf x\in S\}}=p_{u,S}.\]
	Further, for each internal vertex $\mathbf x$ we have, using Lemma \ref{newanticonv}, 
	\begin{align*}\Pr{W_u(\abs{\mathbf x})=\mathbf x}\sqrt{q_{\mathbf{x},S}}
	&=\Pr{W_u(\abs{\mathbf x})=\mathbf x}\sqrt{\operatorname{mC}_2\bigl((q_{\mathbf y,S})_{\mathbf y\in\Gamma^+(\mathbf x)}\bigr)}\\
	&\leq\Pr{W_u(\abs{\mathbf x})=\mathbf x}\sqrt{M_{1/2}\bigl((q_{\mathbf y,S})_{\mathbf y\in\Gamma^+(\mathbf x)}\bigr)}\\
	&=\sum_{\mathbf y\in\Gamma^+(\mathbf x)}\Pr{W_u(\abs{\mathbf x})=\mathbf y}\sqrt{q_{\mathbf{y},S}}.
	\end{align*}
	Summing over all $\mathbf x$ at level $i$, we obtain $\Psi^{(i)}\leq\Psi^{(i+1)}$ for each $i<t$, and consequently $\sqrt{q_{u,S}}=\Psi^{(0)}\leq\Psi^{(t)}=p_{u,S}$, as required.
\end{proof}

\subsection{Random graphs}
We now consider CRW hitting and cover times in the Erd\H{o}s--R\'enyi random graph $\mathcal{G}(n,p)$. This is the probability distribution over all $n$-vertex simple graphs generated by sampling each possible edge independently with probability $p$, see \cite{RGbook} for more details. 
\newtheorem*{thm:gnp}{Theorem \ref{gnp}}
\begin{thm:gnp} Let $\mathcal{G} \overset{d}{\sim}\mathcal{G}(n,p)$ where $  np\geq c\ln n$ for any fixed $c>1$ and $\log np =\lo{\log n}$. Then w.h.p.\ 
		\begin{enumerate}[(i)]
			\item  $\ttwo\left(\mathcal{G}\right) = \BO{n \cdot \log(np)\cdot \log\log n}$ 
			\item  $\ttwoh\left(\mathcal{G}\right) = n^{1-\Omega(1/ \log(np))}$ .
		\end{enumerate} \end{thm:gnp}

\begin{proof}To begin we show that the graph is almost regular w.h.p.
	\begin{claim}
		For $p$ as above, $\dmin,\dmax = \BT{np}$ w.h.p.
	\end{claim}
\begin{poc}In $\mathcal{G} \overset{d}{\sim}\mathcal{G}(n,p)$ since each edge is independent with probability $p$, each degree $d(u)$ is distributed as a binomial random variable $\operatorname{Bin}(n-1,p)$. The Chernoff bound \cite[Thm.\ 3.2]{Conc} states that for any $\lambda$, $\Pr{\operatorname{Bin}(n,p)\geq np +\lambda } \leq \exp\left(-\frac{\lambda^2}{ 2(np + \lambda/3)}\right)$. Thus by a union bound over all vertices $\dmax \leq 5np $ w.h.p.. For $ \dmin$ note that the expected number of vertices of degree $k$ is given by $x_k= n\binom{n-1}{k}p^k(1-p)^{n-1-k}$. We shall consider $k=\kappa np$ for $\kappa \leq 1/2$, in this case $x_{k}/x_{k-1} = \frac{np}{k(1-p)} \geq 2$ and so the expected number of vertices with degree $k$ is $\BO{x_k}$. Observe that $ x_k \leq \left(\frac{npe}{k}\right)^k\exp(\ln n - np) $ and so if $np\geq 5 \log n$ then choosing $k=np/e$ yields $x_k \leq (e^2)^{np/e}\exp(\log n-np) = e^{-\Omega(\log n)}$. Otherwise since $np \geq c\ln n$ where $c>1$ fixed, setting $k = \kappa np $ we have $x_k \leq \left(\frac{ce}{\kappa}\right)^{\kappa \ln n}\exp(- (c-1)\ln n ) $ thus if $\kappa$ satisfies $\kappa\ln(\frac{ce}{\kappa} ) < (c-1)$ then $x_{k}=\lo{1}$. Choosing $\kappa = (c-1)^2/500$ suffices. Since in either case we showed $x_k=\lo{1}$ for some $k= \Omega(np)$ by Markov's inequality $ \dmin = \Omega(np)$ w.h.p. \end{poc} Cooper \& Frieze \cite{CoFrGnp} show that for $np= c\ln n$, $c>1$ w.h.p.\ the conductance of $G(n,p)$ is at least $1/6$, implying that $\trel =\BO{1}$ \cite[Thm.\ 13.14]{levin2009markov}. For larger values of $np$, Coja-Oghlan \cite[Thm.\ 1.2]{CojaEigen} showed that there exists some $c<\infty $ such that for $np \geq c\log n$ the spectral gap of the normalised Laplacian of $\mathcal{G}(n,p)$ is $1-\BO{1/\sqrt{np}}$ w.h.p. Since the normalised Laplacian $\mathcal{L}$ is similar to the random walk Laplacian $\mathcal{L}'$, and the later is given by $\mathcal{L}' =I-P$ we see that also in this range $\trel =\BO{1}$. We have shown that, in this regime, $\mathcal{G}(n,p)$ is almost regular and has constant relaxation time w.h.p., thus $\thit =\BO{n}$ w.h.p.\ by \cite[Thm.\ 5.2]{Chand}. Theorems \ref{trelbdd} \& \ref{trelhit} now yield the results. 
\end{proof}
Thus the CRW gives a significant improvement in the cover and hitting times whenever degrees of $\mathcal{G}(n,p)$ are subpolynomial in $n$. 
	
\section{Computing Optimal Choice Strategies} \label{complexsec}

In this section we focus on the following problem: given a graph $G$ and an objective, how can we compute a series of choices for the walk which achieves the given objective in optimal expected time? In particular we consider the following computational problems related to our main objectives of max/minimising hitting times, cover times and stationary probabilities $\pi_v$.
\begin{labeling}{$\mathtt{Hit}\left(G,v,S\right)$:}
	\item[$\mathtt{Stat}\left(G,w\right)$:] Find a CRW strategy min/maximising $\sum_{v  \in V}w_v \pi_{v}$ for vertex weights $w_v\geq 0$.
	\item[$\mathtt{Hit}\left(G,v,S\right)$:] Find a CRW strategy minimising $\Htwo{v}{S}$ for a given $S\subseteq V(G)$ and $v\in V(G)$.
	\item[$\mathtt{Cov}\left(G,v\right)$:] Find a CRW strategy minimising  $\Ctwo{v}{G}$ for a given $v \in V(G)$.
\end{labeling}
The analogous problems to $\mathtt{Stat}\left(G,w\right)$ and $\mathtt{Hit}\left(G,v,S\right) $  were studied in \cite{ABKLPbias} for the biased random walk. While $\mathtt{Stat}$ is not one of our primary objectives, we include it here both as a natural problem to consider but also because of its relationship to $\mathtt{Hit}$ in the case where $w$ is the indicator function of a set $S$; we shall abuse notation by writing $\mathtt{Stat}(G,S)$ for this case. 
Clearly for $\mathtt{Stat}$ we must restrict ourselves to unchanging strategies for the stationary probabilities $\pi_v$ to be well-defined; we shall show that $\mathtt{Hit}$ also has an unchanging optimal strategy.

For $\mathtt{Hit}$ and $\mathtt{Cov}$, there are two possible interpretations of what it means to ``find'' a CRW strategy. Perhaps the most natural is to compute a sequence of optimal choices in an on-line fashion, that is at each time step to compute which of the two offered choices to accept. For any particular walk, with suitable memoisation, at most a polynomial number of such computations will be required for either problem: which choice to accept depends only on the current vertex, the two choices, and in the case of $\mathtt{Cov}$ the vacant set, which can change at most $n$ times. We might alternatively want to compute a complete optimal strategy in advance; for $\mathtt{Hit}$ this requires only a polynomial number of single-choice computations, but for $\texttt{Cov}$ the number of possible situations our strategy must cover will be exponential. However, we shall show that $\mathtt{Cov}$ is hard even for individual choices.

\subsection{A Polynomial-Time Algorithm for \texttt{Stat} and \texttt{Hit}}\label{polyhit}
First, we show how the (unknown) optimal values $\Htwo xv$ determine an optimal strategy for $\mathtt{Hit}(G,\cdot,v)$. In the following two lemmas we will need to work with a multigraph $F$; in this context the choice offered at each stage is between two random edges from the current vertex. 
\begin{lemma}\label{beta-opt}Let $F$ be a multigraph and fix a vertex $v$. Let $v=v_0,v_1,\ldots$ be an ordering of the vertices such that for all $i<j$ we have $\Htwo {v_i}v\leq\Htwo {v_j}v$. Let $\beta$ be the deterministic unchanging strategy given by $\axyz[\beta]{v_i}{v_j}{v_k}=1$ whenever $j<k$. Then $\beta$ is optimal (among all strategies) for $\mathtt{Hit}(F,x,v)$ for every $x\neq v$, and also for the problem of minimising $\Exu v{\tau_v^+}$.
\end{lemma}
\begin{proof}
	Fix an optimal strategy $\alpha$ for $\mathtt{Hit}(F,x,v)$, and for each $y\in\Gamma(x)$ write $q_y$ for the probability that the first step under this strategy is from $x$ to $y$. Recall that $q_y=\sum_{z\in\Gamma(x)}\frac{2\axyz xyz}{  d(x)^2}$. Now given that the first step is at $y$, an optimal strategy for the remaining steps is precisely an optimal strategy for $\mathtt{Hit}(F,y,v)$, and thus
	\[\Htwo xv=1+\sum_{y\in\Gamma(x)}q_y\Htwo yv.\]
	Suppose there exist $y,z\in\Gamma(x)$ with $\Htwo yv<\Htwo zv$ but $\axyz xyz<1$ at the first step. By instead (at time $1$ only) always choosing $y$ in preference to $z$, the expected hitting time is decreased by $\frac{2}{d(x)^2}(1-\axyz xyz)(\Htwo zv-\Htwo yv)$, a contradiction. Thus we have $\axyz xyz=1$ if $\Htwo yv<\Htwo zv$ and $\axyz xyz=0$ if $\Htwo yv>\Htwo zv$. If $\Htwo yv=\Htwo zv$ then the expected hitting time does not depend on $\axyz xyz$, and so any strategy satisfying these conditions at time $1$, and thereafter following an optimal strategy, is itself optimal. 
	
	It follows by induction that following $\beta$ for $k$ turns and thereafter following $\alpha$ is optimal; since this gives arbitrarily good approximations of the expected hitting time under $\beta$, $\beta$ is itself optimal for $\mathtt{Hit}(F,x,v)$, and, since the definition of $\beta$ does not depend on $x$, for $\mathtt{Hit}(F,y,v)$ for any $y\neq v$. 
	
	Next we show that $\beta$ is also an optimal strategy for minimising $\Exu v{\tau_v^+}$. Suppose not, and let $\gamma$ be an optimal strategy. Write $q_x^{\gamma}$ for the probability of moving from $v$ to $x$ at time $1$ under $\gamma$, and $H_v^{\gamma}(v^+)$ for $\Exu v{\tau_v^+}$ under $\gamma$. Now
	\begin{align*}H_v^{\gamma}(v^+)&=1+ \sum_{x\in\Gamma(v)}q_x^{\gamma}H_x^{\gamma}(v)\\
	&\geq 1 + \sum_{x\in\Gamma(v)}q_x^{\gamma}H_x^{\beta}(v),
	\end{align*}
	by optimality of $\beta$ for $\mathtt{Hit}(F,x,v)$. Suppose $\axyz[\gamma]vxy\neq\axyz[\beta]vxy$ for some $x,y\in\Gamma(v)$. 
	Replacing $\axyz[\gamma]vxy$ and $\axyz[\gamma]vyx$ by $\axyz[\beta]vxy$ and $\axyz[\beta]vyx$ respectively changes
	$\sum_{x\in\Gamma(v)}q_x^{\gamma}H_x^{\beta}(v)$ by $\frac2{d(v)^2}(\axyz[\beta]vxy-\axyz[\gamma]vxy)(\Htwo xv-\Htwo yv)$, 
	which is non-positive by choice of $\beta$. Thus after a sequence of such changes we obtain
	\begin{align*}H_v^{\gamma}(v^+)&\geq 1+\sum_{x\in\Gamma(v)}q_x^{\gamma}H_x^{\beta}(v)\\
	&=H_v^{\beta}(v^+).\qedhere\end{align*}
\end{proof}

\begin{remark}
	In particular, recalling that for an unchanging strategy $\pi_v = 1/\Exu{v}{\tau_v^+}$, it follows that $\beta$ is an optimal strategy for $\mathtt{Stat}(F,\{v\})$. However, this is true in a somewhat stronger sense, since optimality for $\mathtt{Stat}$ only requires minimising $\Exu{v}{\tau_v^+}$ among unchanging strategies, whereas Lemma \ref{beta-opt} shows that $\beta$ minimises this quantity among all strategies; we shall need this extra strength.
\end{remark}
Note that there may be other deterministic unchanging optimal strategies for $\mathtt{Hit}(F,x,v)$. For example, if there are multiple vertices with the same optimal hitting time, we may choose between them arbitrarily, and in particular may have a cyclic order of preference which is not consistent with any single ordering. 
The following lemmas will enable us to show that a good enough approximation to an optimal strategy must itself be optimal.
\begin{lemma}\label{approx-strat}Let $F$ be a multigraph with at most $n$ vertices and at most $\binom n2$ edges, and fix a vertex $v$. Let $\alpha$ be any unchanging strategy for $\mathtt{Stat}(F,\{v\})$. Suppose there exist vertices $x,y,z$ with $y,z\in\Gamma(x)$, $\Htwo yv<\Htwo zv$ and $\axyz xyz\leq1/2$. Then $\pi^{\alpha}_v$ differs from the optimal value by at least $n^{-4(n+1)}(\Htwo zv-\Htwo yv)$.
\end{lemma}
\begin{proof}First we bound $H_v^\alpha(v^+)-H_v^\beta(v^+)$, where $\beta$ is as described in Lemma \ref{beta-opt}. Consider the strategy of following $\alpha$ until the first time the walk either reaches $v$ or is at $x$ and offered a choice between $y$ and $z$, and in the latter case following $\beta$ until $v$ is reached. The difference between this strategy and following $\alpha$ is $p(\axyz xyz H_y^{\alpha}(v)+\axyz xzy H_z^{\alpha}(v)-H_y^{\beta}(v))$, where $p$ is the probability of the latter event occurring before the walk returns to $v$. Note that
	\begin{align*}\axyz xyz H_y^{\alpha}(v)+\axyz xzy H_z^{\alpha}(v)-H_y^{\beta}(v)&\geq(\axyz xyz-1)H_y^{\beta}(v)+\axyz xzy H_z^{\beta}(v)\\
	&=(1-\axyz xyz)(\Htwo zv-\Htwo yv)\\
	&\geq(\Htwo zv-\Htwo yv)/2
	\end{align*}
	by Lemma \ref{beta-opt} and the assumptions. Further, 
	\[p\geq2\bfrac1{\dmax(F)^2}^{d(v,x)+1}\geq\binom n2^{-2n},\]
	since with at least this probability the walk is forced along a specific shortest path to $x$, then offered a choice of $y$ or $z$.
	
	Thus the difference in $\Exu{v}{\tau_v^+}$ between $\alpha$ and this hybrid strategy is at least \[\zeta:=\frac 12\binom n2^{-2n}(\Htwo zv-\Htwo yv),\] and since $\beta$ minimises this quantity among all strategies by Lemma \ref{beta-opt}, the same bound applies to the difference between $\alpha$ and $\beta$, giving \[\pi_\alpha(v)^{-1}\geq \pi_\beta(v)^{-1}+\zeta,\]
	and consequently
	\begin{equation}\label{pi-diff}
	\pi_\alpha(v)\leq \pi_\beta(v)-\zeta\frac{\pi_\beta(v)^2}{1+\pi_\beta(v)\zeta}.\end{equation}
	We have $1\geq \pi_\beta(v)\geq\binom n2^{-1}$ by comparison with a simple random walk. Also we may crudely bound $\ttwoh$ by noting that a SRW has probability at least $\binom n2^{1-n}$ of reaching any given vertex in at most $n-1$ steps, giving $\zeta<1$.
	Combining these bounds with \eqref{pi-diff} gives the required result.
\end{proof}

\begin{lemma}\label{approlem}For any simple graph $G$ of order $n$ and every pair of vertices $x,y$ with $\Htwo{x}{S}<\Htwo{y}{S}$ we have $\Htwo{y}{S}-\Htwo{x}{S}>n^{-2n^2}$.
\end{lemma}

\begin{proof}
	Note that the hitting times $\left( h_x \right)_{x \in V}$ of $S$ from $x$ for any given unchanging strategy are uniquely determined by the equations
	\[h_x=\begin{cases}1+\sum_y \mathbf{P}_{xy}\cdot h_y&\quad\text{if }x\not\in S\\
	0&\quad\text{if }x\in S,\end{cases}\]
	where $\mathbf P$ is the transition matrix for the strategy. This set of equations can be written as $ \mathbf A \mathbf h=\mathbf b$, where $\mathbf A :=\left(\mathbf{I} -\mathbf{Q}\right)$, $\mathbf{Q}_{i,j} =\mathbf{P}_{i,j} $ if $i\notin S$ and $0$ otherwise, and $\mathbf b$ is a $0$-$1$ vector. Notice that $\mathbf{A}$ is diagonally dominant, and from any row where equality occurs there is a path of non-zero entries to a strictly dominant row. It is straightforward to check that such a matrix is invertible: see for example \cite[Lemma 3.2]{AF16}. 
	For any non-random strategy, and in particular for the optimal strategy described above, every transition probability from $x$ is a multiple of $d(x)^{-2}$. Thus all the elements of $\mathbf A$ can be put over a common denominator $D$, where $D:=\operatorname{LCM}(\mdeg{x}^2 )_{x \in V}<(n!)^2<n^{2n}/2$.
	
	We have $\mathbf h = \mathbf A^{-1}\mathbf b$ = $\abs{\mathbf A}^{-1} \mathbf{C}^{\mathsf T} \mathbf b$, where $\mathbf C$ is the matrix of cofactors. Each entry in $\mathbf C$ can be put over a common denominator which is at most $D^n$, and so the same applies to each entry of $\mathbf C^{\mathsf T}\mathbf b$. Also, $\abs{\mathbf A}<2^n$ by Hadamard's inequality \cite[Thm. 7.8.1]{Horn}. It follows that if two hitting times differ, they differ by at least $(2D)^{-n}$.
\end{proof}

For any graph $G$ and weighting $w:V\rightarrow [0,\infty)$ on the vertices of $G$ we can phrase $\mathtt{Stat}\left(G,w\right) $ as an optimisation problem as follows, where we shall encode our actions using the probabilities $\axyz xyz = \Pr{X_{t+1} =y \mid X_t = x, c=\{y,z \} }  $ from Section \ref{formaldef}. 

\begin{equation}\label{statop}
\begin{aligned}
& \text{maximize:}\! & \sum\limits_{v  \in V}w_v \pi(v)&  & \\
& \text{subject to:}\!&  \pi(x)&=   \sum_{y\in \Gamma(x)}\pi(y) \cdot \frac{2\sum_{z \in \Gamma(y)}\axyz yxz}{\mdeg{x}^2}   
,\quad& \forall x\in V \\
&  & \sum_{x\in V}\pi(x)  &=1,  & \\
&    & \axyz xzy  &\in [0,1],  &\forall xz,xy\in E \\
&   & \axyz xzy &= 1-\axyz xyz,  &\forall xz,xy\in E \\
\end{aligned}
\end{equation}For minimising the stationary probabilities we maximise $-1$ times the objective function. 

To prove Theorem \ref{progstat} the quadratic terms in \eqref{statop} can be eliminated using the same substitution as \cite[Theorem 6]{ABKLPbias}; we can then solve \eqref{statop} as a Linear Program.

\begin{theorem} \label{progstat} For any multigraph $G$ and weight function $w:V\rightarrow [0,\infty)$ a policy solving the problem  $\mathtt{Stat}\left(G,w\right) $ to within an additive $\eps$ factor can be computed in time $\poly\!\left(\abs E, \log(1/\eps)\right)$.  
\end{theorem}

\begin{proof}
	We prove the simple graph case; this proof may be easily extended for multigraphs with suitably adapted notation. 
	The optimisation problem \eqref{statop} above can be rephrased as a Linear Program by making the substitution $r_{x,y,z}=\pi(x)\cdot \axyz xyz$. 
	Either the Ellipsoid method or Karmarkar's algorithm will approximate the solution to within an additive $\eps>0$ factor in time which is polynomial in the dimension of the problem and $\log(1/\eps)$, see for example \cite{GroLovSch,Karmarkar}.
\end{proof}

We now show how one can now use this linear program to determine the hitting times.
\newtheorem*{thm:hitexact}{Theorem \ref{hitexact}}
\begin{thm:hitexact}
 For any graph $G$ and any $S \subset V$, a solution to $\mathtt{Hit}\left(G,x,S\right) $ for every $x\in V\setminus S$  can be computed in time $\poly(n)$. 
\end{thm:hitexact}
\begin{proof}
	Contract $S$ to a single vertex $v$ to obtain a multigraph $F$; where a vertex $x$ has more than one edge to $S$ in $G$, retain multiple edges between $x$ and $v$ in $F$. Note that $F$ has at most $n$ vertices and at most $\binom n2$ edges. 
	Provided that the CRW on $G$ has not yet reached $S$, there is a natural correspondence between strategies on $G$ and $F$ with the same transition probabilities, and it follows that $\Htwo xS$ for $G$ and $\Htwo xv$ for $F$ are equal for any $x\in V(G)\setminus S$. We compute an optimal strategy to $\mathtt{Stat}(F,\{v\})$ to within an additive error of $\eps:=n^{-10n^2}$; note that $\log(1/\eps)=o(n^3)$ and so this may be done in time $\poly(n)$ by Theorem \ref{progstat}. 
	Applying Lemma \ref{approx-strat} to $F$ and Lemma \ref{approlem} to $G$, using the equality of corresponding hitting times, implies that this strategy has $\axyz xyz>1/2$ whenever $\Htwo yv<\Htwo zv$, and so rounding each of the probabilities $\axyz xyz$ to the nearest integer gives an optimal strategy (on $F$) for every $x$, which may easily be converted to an optimal strategy for $G$.
\end{proof}

\subsection{A Hardness Result for \texttt{Cov}} 
We show that in general even the on-line version of $\mathtt{Cov}\left(G,v\right)$ is $\NP$-hard. To that end we introduce the following problem, which represents a single decision in the on-line version. The input is a graph $G$, a current vertex $u$, two vertices $v$ and $w$ which are adjacent to $u$, and a visited set $X$, which must be connected and contain $u$.
\begin{labeling}{$\mathtt{NextStep}\left(G,u,v,w,X\right)$:}
	\item[$\mathtt{NextStep}\left(G,u,v,w,X\right)$:] Choose whether to move from $u$ to either $v$ or $w$ so as to minimise the expected time for the CRW to visit every vertex not in $X$, assuming an optimal strategy is followed thereafter.
\end{labeling}
Any such problem may arise during a random walk with choice on $G$ starting from any vertex in $X$, no matter what strategy was followed up to that point, since with positive probability no real choice was offered in the walk up to that point. 

\newtheorem*{thm:CovIsHard}{Theorem \ref{CovIsHard}}
\begin{thm:CovIsHard}$\mathtt{NextStep}$ is $\NP$-hard, even if $G$ is constrained to have maximum degree $3$.\end{thm:CovIsHard}
\begin{proof}
	We give a (Cook) reduction from the $\NP$-hard problem of either finding a Hamilton path in a given graph $H$ or determining that none exists. This is known to be $\NP$-hard even if $H$ is restricted to have maximum degree $3$ \cite{GJS-NPC}.
	
	We shall find it more convenient to work with the following problem, which takes as input a graph $G$, a current vertex $u$ and a connected visited set $X$ containing $u$.
\begin{labeling}{$\mathtt{BestStep}\left(G,u,X\right)$:}
		\item[$\mathtt{BestStep}\left(G,u,X\right)$:] Choose a neighbour of $u$ to move to so as to minimise the expected time for the CRW to visit every vertex not in $X$, assuming an optimal strategy is followed thereafter.
	\end{labeling}
	We may solve $\mathtt{BestStep}(G,u,X)$ by computing $\mathtt{NextStep}(G,u,v,w,X)$ for every pair $v,w$ of neighbours of $u$; since all optimal neighbours must be preferred to all others, this will identify a set of one or more optimal choices for $\mathtt{BestStep}(G,u,X)$.
	Consequently, it is sufficient to reduce the Hamilton path search problem to $\mathtt{BestStep}$.
	
	Given an $n$-vertex graph $H$, construct the graph $G$ as follows. First replace each edge of 
	$H$ by a path of length $2cn^2$ through new vertices. Next add a new pendant path of length $n^3$ 
	starting at the midpoint of each path corresponding to an edge of $H$. Finally, add edges to form a cycle consisting of the end 
	vertices of these pendant paths (in any order). Note that if $H$ has maximum degree $3$, so does~$G$.
	
	Fix a starting vertex $u$ and a non-empty unvisited set $Y\subseteq V(H)\setminus\{u\}$, and set $X=V(G)\setminus Y$. (The purpose of the second and third stages of the construction is to make $X$ 
	connected without affecting the optimal strategy.) Suppose that $H$ contains at least one path of 
	length $\abs{Y}$ starting at $u$ which visits every vertex of $Y$; in particular if $Y=V(H)\setminus\{u\}$ this is a Hamilton path of $H$. We claim that any optimal next step is to move towards the next vertex 
	on some such path. Assuming the truth of this claim, an algorithm to find a Hamilton path starting at 
	$x$, if one exists, is to set $u=x$ and $Y=V(H)\setminus\{x\}$, then find the vertex $y$ such that 
	moving towards $y$ is optimal, set $u=y$ and remove $y$ from $Y$, then continue. If this fails to find 
	a Hamilton path, repeat for other possible choices of $x$.
	
	To prove the claim, first we argue by induction that there is a strategy to visit every vertex in $\abs{Y}$ in expected time $(4cn^2+\BO{n})\abs{Y}$, where the implied constant does not depend on $c$. 
	This is clearly true for $\abs{Y}=0$. Let $y$ be the next vertex on a suitable path in $H$, and let 
	$z$ be the middle vertex of the path corresponding to the edge $uy$. Attempting to reach $z$ by a 
	straightforward strategy, the distance to $z$ evolves as a random walk with probability $3/4$ of 
	decreasing unless the current location is a branch vertex. We thus reach $z$ in expected time $2cn^2$ 
	plus an additional constant time for each visit to $u$, of which we expect $\BO{d(u)}=\BO{n}$, giving a 
	total expected time of $2cn^2+\BO{n}$ (if the walker is forced to a different branch vertex first, the 
	expected time to return from this point is polynomial in $n$, but this event occurs with exponentially 
	small probability). Similarly, the time taken to reach $y$ from $z$ is $2cn^2+\BO{1}$. Once $y$ is 
	reached, there is (by choice of $y$) a path of length $\abs{Y}-1$ in $H$ starting from $y$ and
	visiting all of $Y\setminus\{y\}$. Thus, by induction, the required bound holds.
	Secondly, suppose that an optimal first step in a strategy from $u$ moves towards a vertex $y'$ of $H$ 
	which is not the first step in a suitable path. Since the expected remaining time decreases whenever 
	an optimal step is taken, two successive optimal steps cannot be in opposite directions unless the
	walker visits a vertex of $Y$ in between. Thus the optimal strategy is to continue in the direction of 
	$y'$ if possible, and such a strategy reaches $y'$ before returning to $u$ with at least constant 
	probability $p$, and this takes at least $2cn^2$ steps. Note that the expected time taken to reach 
	another vertex of $H$ from a vertex in $H$, even if the walker is purely trying to minimise this 
	quantity, is at least $4cn^2$, and from either $u$ or $y'$ at least $\abs{Y}$ such transitions are 
	necessary to cover $Y$. Thus such a strategy, conditioned on the first step being in the direction of 
	$y'$, has expected time at least $4cn^2+2pcn^2$, which, for suitable choice of $c$, proves the claim.
\end{proof}

\subsection{Computing \texttt{Cov} via Markov Decision Processes}
To compute a solution for $\mathtt{Cov}\left(G,v\right) $ we can encode the cover time problem as a hitting time problem on a (significantly) larger graph.
\begin{lemma}\label{covashit}
	For any graph $G=(V,E)$ let the (directed) auxiliary graph $\tilde{G}=(\tilde{V},\tilde{E})$ be given by $\tilde{V}=\{(v,S):  S\subseteq V, v\in S\}$ and $\tilde{E}= \left\{((i,S),(j,S\cup j))\mid ij \in E\right\}$. Then solutions to $\mathtt{Cov}\left(G,v\right)$ correspond to solutions to $\mathtt{Hit}\bigl(\tilde{G},\tilde{v},W\bigr)$ and vice versa, where $W=\{(u,V)\mid u\in V\}$.
\end{lemma}

\begin{proof}
	There is a natural bijection between the out-edges in $G$ from $u$ and those in $\tilde G$ from $(u,T)$ for any $u\in V,T\subseteq V$. This extends to a natural bijection from finite walks (which we may think of as a vertex together with a history) in $G$ starting from $v$ to walks in $\tilde G$ starting from $\tilde v$, and also to a measure-preserving bijection between the choices which may be offered from $u$ and $(u,T)$. Thus there is a natural bijection between strategies for the two walks, and both the choices offered and any random bits used may be coupled so that corresponding strategies produce corresponding walks. Since the walk in $G$ has covered $V$ if and only if the walk in $\tilde G$ has hit some vertex in $W$, the times that these events first occur are identically distributed for corresponding strategies, and in particular the sets of optimal strategies correspond.
\end{proof}

In light of Lemma \ref{covashit} it may appear that we can solve $\mathtt{Cov}(G,v)$ by converting it to an instance of $\mathtt{Hit}\bigl(\tilde{G},\tilde{v},W\bigr)$ and appealing to Theorem \ref{hitexact}. This is unfortunately not the case as $\tilde{G}$ is a directed graph and Theorem \ref{hitexact} cannot handle directed graphs. Lemma \ref{covashit} is still of use as we can phrase $\mathtt{Hit}$ in terms of Markov Decision processes and then standard results tell us that an optimal strategy for the problem can be computed in finite time. 

A Markov Decision Process (MDP) is a discrete time finite state stochastic process controlled by a sequence of decisions \cite{Derman}. At each step a controller specifies a probability distribution over a set of actions which may be taken and this has a direct affect on the next step of the process. Costs are associated with each step/action and the aim of the controller is to minimise the total cost of performing a given task, for example hitting a given state. In our setting the actions are orderings of the vertices in each neighbourhood and the cost of each step/action is one unit of time. The problem $\mathtt{Hit}\bigl(G,u,v)$ is then an instance of the optimal first passage problem which is known to be computable in finite time \cite{Derman}. 

\begin{corollary}\label{covalg}For any graph $G$ and $v\in V$ an optimal policy for the problem  $\mathtt{Cov}\left(G,v\right) $ can be computed in exponential time. 
\end{corollary} 

\begin{proof}We first encode the problem $\mathtt{Cov}\left(G,v\right)$ as the problem $\mathtt{Hit}\bigl(\tilde{G},\tilde{v},W\bigr)$ as described in Lemma \ref{covashit}. Now as mentioned $ \mathtt{Hit}\bigl(\tilde{G},\tilde{v},W\bigr)$ is an instance of the optimal first passage problem which for a given graph $\tilde{G}$, start vertex $\tilde{v}$ and target vertex $W$ can be computed in finite time using either policy iteration or linear programming, see for example \cite[Ch.\ 5 Cor.\ 1]{Derman}. Examination of the linear program on \cite[page.\ 58]{Derman} reveals that there is a constraint for every ordering of the neighbours of each vertex. Since $\tilde{G}$ has at most $2^n$ vertices and each of these has at most $n$ neighbours we see that there are at most $2^n\cdot n!\leq e^{n^3}$ constraints. It follows that this Linear program can be solved in time $\poly(e^{n^3})$ thus $\mathtt{Cov}\left(G,v\right) \in \EXP$.
\end{proof}

\begin{remark}Since in our setting actions are orderings of neighbourhoods the space of actions may be factorial in the size of the graph. The algorithms for computing $\mathtt{Hit}\bigl(G,u,v)$ from \cite{Derman} used to establish Corollary \ref{covalg} are polynomial in the number of actions and thus will not yield a polynomial time algorithm for the problem. This is why we resisted appealing to MDP theory when finding a polynomial time algorithm for $\mathtt{Hit}\bigl(G,u,v)$ on undirected graphs in Section \ref{polyhit}. \end{remark}

\section{Summary}

In this paper we proposed a new random walk process inspired by the power of choice  paradigm. We derived several quantitative bounds on the hitting and cover times, and also presented a  dichotomy with regards to computing optimal strategies. 

While we were able to show that on an expander graph, the CRW significantly outperforms the simple random walk in terms of its cover time, we do not yet know the exact order of magnitude of $\ttwo$. In fact, we do not have any lower bound on $\ttwo$ improving the trivial $\Omega(n)$ for any sequence of bounded degree graphs. Constructing a sequence of graphs $(G_n)$, especially expanders, with $\ttwo(G_n) = \omega(n)$ would be very interesting.

We have shown that $\mathtt{Cov}\in \EXP $ and that the problem is $\NP$-hard. It would be interesting to find a complexity class for which the problem is complete, and we suspect it is $\PSPACE$-complete.

\section*{Acknowledgements}
We thank Sam Olesker-Taylor for many insightful comments on an earlier version of this work. A.G. and J.H. were supported by ERC Starting Grant no.\ 639046 (RGGC). J.H. was also partially supported by UK Research and Innovation Future Leaders Fellowship MR/S016325/1. T.S. and J.S. were supported by ERC Starting Grant no.\ 679660 (DYNAMIC MARCH).

\bibliographystyle{abbrv}

\end{document}